\documentclass{aptpub}

\authornames{Jiaming Xu, Yuan Zhong} % insert the authors here for use in running head
\shorttitle{Queue-Size Scaling} % insert short title here for use in running head

\usepackage{graphicx,epsfig,graphics,epsf}
\usepackage{comment, color, caption, subcaption}
\usepackage{soul}
\usepackage{bm}

\usepackage{xspace}

\newcommand{\veps}{\varepsilon}

\newcommand{\EE}{\mathbb{E}}
\newcommand{\PP}{\mathbb P}

\newcommand{\bOne}{\boldsymbol{1}}

\newcommand{\NN}{\mathbb{N}}

\newcommand{\bsig}{\bm{\sigma}}
\newcommand{\cS}{\mathcal{S}}
\newcommand{\bQ}{\bm{Q}}

\newcommand{\ZZ}{\mathbb{Z}}
\newcommand{\bq}{\bm{q}}
\newcommand{\bg}{\bm{g}}

\newcommand{\cR}{\mathcal{R}}
\newcommand{\cC}{\mathcal{C}}
\newcommand{\udl}{\underline}

\newcommand{\bA}{\bm{A}}
\newcommand{\cE}{\mathcal{E}}

\newcommand{\cW}{\mathcal{W}}

\newcommand{\ER}{Erd\H{o}s-R\'{e}nyi\xspace}

\begin{document}

\title{Improved Queue-Size Scaling for Input-Queued Switches via Graph Factorization}

\authorone[Fuqua School of Business, Duke University]{Jiaming Xu} 
% Affiliation is just the name of your university or institution
\authortwo[Booth School of Business, University of Chicago]{Yuan Zhong}

\addressone{100 Fuqua Drive, Durham NC 27708, USA. Email: jx77@duke.edu}

\addresstwo{5807 South Woodlawn Ave, Chicago, IL 60637, USA. Email: yuan.zhong@chicagobooth.edu}
\begin{abstract}
\footnote{This paper makes extensive use of the asymptotic notation. 
For completeness and easy reference, we provide their definitions here. 
Consider a positive parameter $r > 0$, and let $g(r)$ and $f(r)$ be two functions 
of $r$ that take positive real values. Then, $f(r) = O\left(g(r)\right)$ if 
$\limsup_{r\to\infty} f(r)/g(r) < \infty$, $f(r)=\Omega\left(g(r)\right)$ if $\liminf_{r\to\infty} f(r)/g(r) > 0$, 
$f(r) = \omega\left(g(r)\right)$ if $\liminf_{r\to\infty} f(r)/g(r) = \infty$, 
and $f(r) = \Theta\left(g(r)\right)$ if $f(r) = O(g(r))$ and $f(r) = \Omega(g(r))$. 
In our applications, the asymptotic relations may often appear to involve two parameters, 
such as the case of $O\left(n(1-\rho)^{-1}\right)$, where $n$ and $\rho$ are parameters. 
In these cases, simply treat the function of these two parameters as a single parameter; 
i.e., let $r = n(1-\rho)^{-1}$, and the asymptotic notation is well-defined.}
This paper studies the scaling of the expected total queue size in an $n\times n$
input-queued switch, as a function of both the load $\rho$ and the system scale $n$.
We provide a new class of scheduling policies under 
which the expected total queue size scales as 
$O\left( n(1-\rho)^{-4/3} \log \left(\max\{\frac{1}{1-\rho}, n\}\right)\right)$, over all $n$ and $\rho<1$, 
when the arrival rates are uniform. This improves
over the  previously  best-known scalings in two regimes: $O\left(n^{1.5}(1-\rho)^{-1} \log \frac{1}{1-\rho}\right)$
when $\Omega(n^{-1.5}) \le 1-\rho \le O(n^{-1})$ and $O\left(\frac{n\log n}{(1-\rho)^2}\right)$ 
when $1-\rho \geq \Omega(n^{-1})$.
 A key ingredient in our method
is a tight  characterization  of the largest $k$-factor of a random bipartite multigraph, which may be of independent interest. 
\end{abstract}

%\maketitle

%\renewcommand{\thefootnote}{\arabic{footnote}}
%\setcounter{footnote}{0}

\section{Introduction}\label{sec:intro}
An $n\times n$ input-queued switch is a discrete-time queueing network that consists of $n^2$ queues, 
which are arranged in the form of an $n\times n$ matrix. 
Packets arrive to these $n^2$ queues exogenously according to independent Bernoulli processes. 
In each time slot, 
packets are processed according to {\em schedules}, subject to the following constraints: 
at most one queue can be served from each row and  each column, and when a queue is chosen 
to be served, at most one packet can be processed. After the packets are processed, they immediately depart the system. 
The main goal  is to design a scheduling policy 
that decides which queues to serve in each time slot, in order to optimize certain performance measures. 

The input-queued switch plays a pivotal role in the study of  data packets scheduling 
in an internet router \cite{MAW96} and more recently in a data center network \cite{pfabric, coflow, fastpass}. 
It also serves as a prominent example of ``stochastic processing networks'' (SPN) \cite{harrison:canonical, harrison:canonical.corr}, 
which is a canonical model of dynamic resource allocation 
problems that arise in a wide range of applications. 
The study of the switch has paved the way towards a better understanding of more general SPNs; 
indeed, insights generated from policy design and performance analysis 
in a switch often extend to broader classes of SPNs \cite{LD05, MAW96, SWZ11, Stolyar2004, TE92}. 

Even though basic performance measures
such as throughput and stability, are relatively well understood for switches (and for general SPNs) 
(see e.g., \cite{daibala, KW04, K-M, LD05, MAW96, TE92, DLK01}), 
there is a wide gap in our understanding of more refined performance measures, 
such as moments of queue sizes. 
In this paper, we reduce this gap by designing new scheduling policies that achieve tighter bounds 
on the expected total queue size in switches. 
Our study is motivated by \cite{STZopen}, which stated the conjecture that 
the expected total queue size scales as $\Theta \left(\frac{n}{1-\rho}\right)$  
in an $n\times n$ switch, where $\rho \in (0,1)$ is the load factor 
and equals the largest total arrival rate to each row/column of the switch.

It can be easily seen that $\Omega \left(\frac{n}{1-\rho}\right)$ is a lower bound 
on the expected total queue size, over all $n$ and $\rho$ (see e.g., \cite{STZopen}). 
However, matching upper bounds have been only established 
for restricted parameter regimes. 
\cite{NMC07} proposed a batching policy under which the expected total queue size 
is upper bounded by $O\left(\frac{n\log n}{(1-\rho)^2}\right)$, so that when $\rho \in (0, 1)$ is 
treated as a constant, this upper bound matches the lower bound up to a logarithmic factor in $n$. 
\cite{SWZ11} proposed a policy under which the switch emulates a product-form queueing network, 
and established an upper bound of $O\left(\frac{n}{1-\rho}+n^3\right)$,
which matches the lower bound up to a constant factor when $1-\rho = O\left(n^{-2}\right)$.
More recently, a significant work \cite{MS16} shows that the celebrated Max-Weight policy achieves an upper bound of 
$O\left(\frac{n}{1-\rho} + \frac{n^{5-1/r}}{(1-\rho)^{1/r}}\right)$ for any integer-valued $r \geq 2$, 
for systems with uniform arrival rates, which matches the lower bound up to a constant factor 
when $1-\rho = O\left(n^{-a}\right)$ for any $a > 4$. 
Let us also note that using a simple Lyapunov function argument, 
it can be shown that the Max-Weight policy achieves 
an upper bound of $O\left(\frac{n^2}{1-\rho}\right)$ (see e.g., \cite{STZopen, MAW96, TE92, shahkopi}).  

Observe that when $1-\rho = \Theta\left(n^{-1}\right)$, 
the conjectured scaling is $\Theta\left(n^2\right)$, 
whereas all of the aforementioned upper bounds are at least a multiplicative factor $\Omega(n)$ away. 
Motivated by this observation, \cite{STZ16} considered the regime 
where $1-\rho \leq O\left(n^{-1} \right)$, and proposed a scheduling policy 
that gives an upper bound of $O\left(n^{1.5}(1-\rho)^{-1} \log \frac{1}{1-\rho}\right)$, 
for systems with uniform arrival rates.
When $1-\rho = \Theta\left(n^{-1}\right)$, this upper bound becomes $O\left(n^{2.5}\log n\right)$, 
a multiplicative factor of $O\left(n^{0.5}\log n\right)$  away from  the conjectured scaling. 
While this was a significant improvement,  whether this gap 
can be further reduced remains elusive. 

The main contribution of this paper is a new class of policies that achieves an upper bound of 
$O\left(n (1-\rho)^{-4/3} \log \left(\max\{n, 1/(1-\rho)\}\right)\right)$ uniformly over all $n$ and $\rho < 1$, 
when the arrival rates are equal. 
When $\Omega(n^{-1.5}) \le 1-\rho \le O(n^{-1})$, this upper bound improves
over the  previously  best-known scaling $O\left(n^{1.5}(1-\rho)^{-1} \log \frac{1}{1-\rho}\right)$ \cite{STZ16}. 
In particular, when $1-\rho = \Theta\left(n^{-1} \right)$, our upper bound reduces to $O\left(n^{7/3} \log n\right)$, 
an improvement over the previously best-known scaling $O\left(n^{2.5}\log n\right)$. 
When $1-\rho \geq \Omega(n^{-1})$, our upper bound improves over the previously best-known scaling 
$O\left(\frac{n\log n}{(1-\rho)^2}\right)$ \cite{NMC07}. 
Finally, we would also like to point out that besides the upper bound $O\left(\frac{n\log n}{(1-\rho)^2}\right)$ in \cite{NMC07}, 
ours is the only other upper bound that scales almost linearly in $n$, the system scale, 
for any fixed value of $\rho \in (0, 1)$.
See Table \ref{tab:bounds} for a summary of the best-known scalings on the expected total queue size 
under various regimes. 

\begin{table}[h]
\begin{center}
\caption{Best known scalings of the expected total queue size in various regimes. 
Here, $\rho$ is the load factor, and $n$ is the system scale.}
\vspace{.5cm}
\begin{tabular}{|c|c|c|}
\hline
Regime & Scaling & References \\
\hline
$\frac{1}{1-\rho} < n$ & $O\left(\frac{n \log n}{(1-\rho)^{4/3}}\right)$ & this work 
 \\
\hline
$\frac{1}{1-\rho} = n$ & $O\left(n^{7/3}\log n\right)$ & this work \\
\hline
$n\leq \frac{1}{1-\rho} < n^{1.5}$ & $O\left(\frac{n\log n}{(1-\rho)^{4/3}}\right)$ & this work \\
\hline
$n^{1.5} \leq \frac{1}{1-\rho} < n^{2}$ & $O\left(\frac{n^{1.5}\log n}{1-\rho}\right)$ & \cite{STZ16} \\
\hline
$\frac{1}{1-\rho} \geq n^2$ & $\Theta\left(\frac{n}{1-\rho}\right)$ & \cite{SWZ11}\\
\hline
$\frac{1}{1-\rho} \geq n^{4+\veps}$ & $\Theta\left(\frac{n}{1-\rho}\right)$ & \cite{SWZ11, MS16}\\
\hline
\end{tabular}
\label{tab:bounds}
\end{center}
\end{table}

At the heart of our new policy is an efficient scheduling mechanism which depletes the 
packets in the $n^2$ queues as much as possible without any waste of service opportunities. 
Mathematically, this can be formulated as an integer linear optimization problem: 
call a matrix $\bq=(q_{ij})_{i,j=1}^n$ with non-negative integer entries a {\em queue matrix}; 
given $\bq$, find the maximum integer $k$ and a queue matrix $\bg=(g_{ij})_{i,j=1}^n$  
such that  $\sum_{i} g_{ij}=\sum_{j} g_{ij}=k$
and $0 \le g_{ij} \le q_{ij}$. 
Such $\bg$ can then be written as the sum of $k$ maximal schedules, 
which can be used to deplete the $n^2$ queue sizes of $\bq$.
 By representing the rows as left vertices of a bipartite (multi)graph, 
and the columns as right vertices, the queue size $q_{ij}$
can be viewed as the number of edges between left vertex $i$ and right vertex  $j$. 
In this vein, this optimization problem is equivalent to finding the largest $k$-factor (spanning $k$-regular  subgraph)
in a bipartite (multi)graph $G$. Under the Bernoulli arrival assumption, the queue sizes $q_{ij}$ are independently and binomially distributed
for some  common parameters $m$ and $p$. 
We show that if $pmn \ge 152 \log n$, then $G$  has a $k$-factor with 
$k=\lfloor pmn-  \sqrt{304pmn\log n} \rfloor $ with probability $1-n^{-16}$. 
This result is new, and essentially tight: for example, 
when $G$ is an \ER bipartite simple graph, i.e., when $m=1$, then the condition $np=\omega(\log n)$ implies that 
$G$ has a $k$-factor with  $k = np - \Theta(\sqrt{np \log n})$, which is asymptotically the best possible, because the minimum degree of 
$G$ is asymptotically equal to $np-\Theta(\sqrt{np \log n})$ with high probability.

We point out that there is a vast literature on graph factors and factorization (see the survey~\cite{plummer2007graph} and 
the  book~\cite{akiyama2011factors} and the references therein).
%A $k$-factor of a graph $G$ is a spanning $k$-regular subgraph. 
%In particular, $1$-factor is also known as full (perfect) matching. 
%The seminal work~\cite{tutte1952factors} by Tutte gave a necessary and
%sufficient condition for the existence of a $k$-factor in any graph $G$.
%In the context of bipartite graphs, 
%a necessary and
%sufficient condition for the existence of a $k$-factor is given by Gale-Ryser theorem~\cite{gale1957theorem,ryser1957combinatorial},
%which generalizes the celebrated Hall's Theorem~\cite{hall1935representatives} from $k=1$ to  the general $k \ge 1$.
Most previous work focuses on deriving sufficient conditions for the existence of a $k$-factor for  worst-case possible graphs $G$. 
For instance, Csaba~\cite{csaba2007regular} 
showed that for any bipartite simple graph 
$G$ with $n$ left and $n$ right vertices, if the minimum degree
$\delta n \ge n/2$, then $G$ has a $k$-factor with $k=\lfloor n \frac{\delta+\sqrt{2\delta-1}}{2}\rfloor$.
Applying this result to an \ER bipartite simple graph, 
we can only conclude with high probability 
the existence of $k$-factor with $k \sim n \frac{p+\sqrt{2p-1}}{2}$, 
which is suboptimal, when compared to our result.

Random graph factorization was studied earlier in Shamir and Upfal~\cite{shamir1981factors}. It was shown that
for an \ER random graph $G$,
if the average degree $np$ satisfies $np-\log n - (k-1) \log \log n \to +\infty$ for a fixed $k \ge 1$, 
then $G$ has a $k$-factor with probability converging to $1$. However, this result only holds for a fixed $k$ and does not 
provide the largest possible value of $k$ in the dense graph regime where $np=\omega(\log n)$. 
A related random capacitated transportation problem was studied in~\cite{hassin1988probabilistic}, 
where we observe a
$n\times n$ capacity matrix $C=(c_{ij})$ with random $c_{ij}$ 
and aim to find a matrix $X=(x_{ij})$ with the row sums and column sums  as large as possible 
under the constraint that $0 \le x_{ij} \le c_{ij}$.
  Their general result
allows for $c_{ij}$ to be not independent and not identically distributed. 
In the special case where $c_{ij}$ is Bernoulli with the success probability $p$, their result implies that with high probability an 
 \ER random bipartite graph $G$ has $k$-factor with $k \sim n^{1/3}$, which is much smaller than $np$ for large $n$ 
 and hence highly suboptimal.

Before proceeding, let us make two remarks regarding the proposed policy and its performance scaling. 
First, our policy and current analysis rely crucially on the assumption of uniform arrival rates, while some existing
policies and results such as the standard batching policy \cite{NMC07} and the Max-Weight policy \cite{MAW96}
apply to heterogeneous arrival rates. 
Second, our policy is required to know the arrival rates {\em a priori}. 
In contrast, some existing policies, such as Max-Weight policy, are based solely on the observed queue sizes. 

\subsection{Organization}\label{ssec:organize}
The rest of the paper is organized as follows.
The model is formally introduced in Section \ref{sec:model}. 
Then, we present our main result and an overview of the policy in Section \ref{sec:main}. 
Section \ref{sec:prelim} summarizes a few important preliminary results to be used for later parts of the paper. 
Section \ref{sec:envelope} introduces the concept of a lower envelope of a queue matrix and
provides a characterization of the tightest lower envelope of 
a random queue matrix, or equivalently, the largest $k$-factor of a random bipartite multigraph. 
In Section \ref{sec:policy}, we formally describe our policy, and in Section \ref{sec:analysis}, 
we prove the main result. We conclude the paper with some discussion in Section \ref{sec:discussion}.

\subsection{Notation}\label{ssec:notation} 
We reserve boldface letters for vectors and matrices, 
plain and lowercase letters for deterministic scalars,  
and upper-case letters for random quantities. 
We also reserve script letters, such as $\cS$, $\cE$, etc, for sets and/or events. 
For a set $\cS$, we use $\cS^c$ to denote the complement of $\cS$, 
and $|\cS|$ to denote the cardinality of the set $\cS$. 
The indicator function of set $\cS$ is denoted by $\bOne_{\cS}$.  
We use $\NN$ to denote the set $\{1, 2, \cdots\}$ of natural numbers, 
$\ZZ$ to denote the set of integers, and $\ZZ_+$ to denote the set $\{0, 1, 2, \cdots\}$ 
of non-negative integers. 
For $n \in \NN$, we use $[n]$ to denote the set $\{1, 2, \cdots, n\}$. 
For two real numbers $x$ and $y$, use $x \vee y = \max\{x, y\}$ 
to denote the larger of $x$ and $y$. 
For a real number $x$, we use $(x)^+ = x\vee 0$ 
to denote the non-negative part of $x$. 
The shorthand {\em i.i.d} means ``independently and identically distributed'', 
WLOG means ``without loss of generality'', LHS means ``left-hand side'' 
and RHS means ``right-hand side.'' 
We say a sequence of events $\cE_n$ indexed by a positive integer $n$ 
holds with high probability, if the probability of $\cE_n$ converges to $1$ as $n \to +\infty$.
Finally, we note that  notation in this paper
will be made as consistent with that in \cite{STZ16} as possible.

%\begin{lemma}
%\end{lemma}
%
%\begin{proof} 
%
%\end{proof}

\section{Input-Queued Switch Model}\label{sec:model} 
Here we only provide a mathematical description of the model; 
see e.g., \cite{MAW96} for more details on the physical architecture. 
Abstractly, an $n\times n$ input-queued switch consists of $n^2$ queues, 
indexed by the pair $(i, j)$, where $i, j \in [n]$. 
In this paper, the first coordinate $i$ of the pair $(i,j)$ is called an {\em input}, 
the second coordinate $j$ an {\em output}, 
and $n$ is often referred to as the {\em system scale} of the switch. 
The system operates in discrete time, indexed by $\tau \in \NN$. 
For each $i$ and $j$, {\em packets} arrive in queue $(i, j)$ according to some exogenous process. 
Let $A_{ij}(\tau)$ denote the cumulative number of arrivals to queue $(i, j)$ from time slot $1$ to $\tau$. 
Similar to \cite{STZ16}, we assume that arrivals to each of the $n^2$ queues form an independent Bernoulli process 
with rate $\rho/n$, so that $A_{ij}(\tau)-A_{ij}(\tau-1)$ are independent Bernoulli random variables 
with mean $\rho/n$, for $\tau \in \NN$ and $i, j \in [n]$, with the convention that $A_{ij}(0) = 0$ 
for all $i$ and $j$. 
We are only interested in systems that can be made stable under some scheduling policy, 
so we assume that $\rho < 1$, i.e., the system is underloaded. 
%Furthermore, throughout the paper, we consider the scaling regime 
%where $\rho = 1 - 1/f_n$, with $f_n \geq n$ for all $n$.

In each time slot, the switch can serve a number of packets using a {\em schedule}, 
subject to the following constraints: for each $i$, at most one packet can be served 
from any of the $n$ queues $(i, 1), (i, 2), \cdots$, and $(i, n)$; 
and for each $j$, at most one packet can be served from any of 
the $n$ queues $(1, j), (2, j), \cdots$, and $(n, j)$. 
More precisely, let $\bsig = (\sigma_{ij})_{i,j=1}^n$ be a schedule, 
where $\sigma_{ij}$ denotes the number of packets can be served from queue $(i, j)$ in a time slot. 
Then, the set $\cS$ of feasible schedules is given by
\begin{equation}\label{eq:schedule}
\cS = \left\{ \bsig \in \{0, 1\}^{n\times n} : \forall i, \sum_{j'=1}^n \sigma_{ij'} \leq 1; \ \text{ and } \ \forall j, \sum_{i'=1}^n \sigma_{i'j} \leq 1 \right\}.
\end{equation}
Since the set of schedules $\cS$ corresponds exactly to the set of (partial) matchings 
between the inputs and outputs, a schedule is sometimes called a {\em matching} as well. 
A schedule $\bsig = (\sigma_{ij}) \in \cS$ with $1 =  \sum_{j'=1}^n \sigma_{ij'} = \sum_{i'=1}^n \sigma_{i'j}$ 
for all $i$ and $j$ is also called a {\em full (perfect) matching}.
A scheduling {\em policy} decides, in each time slot, the schedule to use for serving the packets in the system, 
based on the past history and the current queue sizes. 

To specify the dynamics of the model, we adopt the following timing convention. 
Let $Q_{ij}(\tau)$ be the size of the queue $(i, j)$ at the beginning of time slot $\tau$. 
In each time slot $\tau$, the policy first observes the queue sizes $Q_{ij}(\tau)$, 
and then decides on the schedule $\bsig(\tau)$, which is applied in the middle of the time slot. 
At the end of the time slot, new arrivals take place. Thus, for each $i$, $j$, and $\tau$, we have
\begin{equation}\label{eq:dynamics}
Q_{ij}(\tau+1) = \left(Q_{ij}(\tau) - \sigma_{ij}(\tau)\right)^+ + A_{ij}(\tau) - A_{ij}(\tau-1).
\end{equation}
Since $\sigma_{ij}(\tau) \in \{0, 1\}$, Eq. \eqref{eq:dynamics} can be rewritten as 
\begin{equation}\label{eq:dynamics1}
Q_{ij}(\tau+1) = Q_{ij}(\tau) - \sigma_{ij}(\tau)\bOne_{\{Q_{ij}(\tau)>0\}} + A_{ij}(\tau) - A_{ij}(\tau-1).
\end{equation}
Throughout the paper, we assume that the system starts empty, i.e., $Q_{ij}(1) = 0$ for all $i$ and $j$. 
Then, similar to \cite{STZ16}, we sum Eq. \eqref{eq:dynamics1} over time to get that for $\tau \in \NN$, 
\begin{equation}\label{eq:dyn_cumm}
Q_{ij}(\tau+1) = A_{ij}(\tau) - S_{ij}(\tau)
\end{equation} %\jx{Shouldn't be $Q_{ij}(\tau+1) = A_{ij}(\tau) - S_{ij}(\tau)$?}
for all $i$ and $j$, where 
\begin{equation}\label{eq:offered_service}
S_{ij}(\tau) = \sum_{t=1}^{\tau} \sigma_{ij}(\tau)\bOne_{\{Q_{ij}(\tau)>0\}}
\end{equation}
is the {\em actual} service received by queue $(i, j)$ from time slot $1$ to $\tau$. 
Note $S_{ij}(\tau)$ is different from $\sum_{t=1}^{\tau} \sigma_{ij}(\tau)$, 
the cumulative service {\em offered} to queue $(i, j)$ from slot $1$ to $\tau$. 
Whenever service is offered for queue $(i, j)$ (i.e., $\sigma_{ij}(\tau) = 1$), but there is no packet to be served 
(i.e., $Q_{ij}(\tau) = 0$), we called the offered service {\em wasted}.

\section{Main Result and Policy Overview}\label{sec:main} 
In this section, we present the main result of the paper, as well as an overview of our policy. 

\subsection{Main result}\label{ssec:main}
%We first formally state our main results, which provide an improved upper bound to the expected total queue size in an
%$n\times n$ input-queued switch. 
\begin{theorem}\label{theorem:main} 
Let $n \in \NN$, and consider an $n\times n$ input-queued switch for which the $n^2$ arrival streams 
form independent Bernoulli processes with a common arrival rate $\rho/n$, where $\rho \in (0, 1)$. 
Then, there exists a scheduling policy under which 
the expected total queue size is upper bounded by $cn (1-\rho)^{-4/3} \log f$ at all times, 
where 
\begin{equation}\label{eq:f}
f = n \vee \frac{1}{1-\rho}, 
\end{equation}
and $c$ is a constant that does not depend on $n$ or $\rho$; i.e., for any $\tau \in \NN$, 
\begin{equation}\label{eq:main}
\EE\left[\sum_{i,j=1}^n Q_{ij}(\tau)\right] \leq cn (1-\rho)^{-4/3} \log f.
\end{equation}
\end{theorem}
%\jx{Shall we emphasize that our scheduling policy is computationally efficient (run time polynomial in $n$)?}
%\jx{I am curious why we must assume $f_n \geq n$? I can see that this assumption is used
%later in the proof, but I do not see why this assumption is fundamental. 
%If it is indeed inevitable, I think we should add a remark explaining why the assumption
%is needed.}

It is worth noting that the bound in \eqref{eq:main} holds 
uniformly over all $n$ and $\rho$; this is different from the upper bound in \cite{STZ16}, 
for example, which requires a joint scaling of parameters $n$ and $\rho$.  
Let us also note that the upper bound in \eqref{eq:main} 
improves over the previously best-known upper bound 
 $O\left(n^{1.5}(1-\rho)^{-1} \log \frac{1}{1-\rho}\right)$ \cite{STZ16}
 when $n \leq \frac{1}{1-\rho} \leq n^{1.5}$, 
 and the previously best-known upper bound $O\left(\frac{n\log n}{(1-\rho)^2} \right)$ 
 when $\frac{1}{1-\rho} < n$.
In particular, when $\frac{1}{1-\rho} = n$, our upper bound improves
the state-of-the-art from $O\left(n^{2.5}\log n\right)$ to $O\left(n^{7/3}\log n\right)$, as 
shown in the following corollary. 

\begin{corollary}\label{corollary:main}
Let $n \in \NN$, and consider an $n\times n$ input-queued switch for which the $n^2$ arrival streams 
form independent Bernoulli processes with a common arrival rate $\rho/n$, where $\rho = 1-1/n$. 
Then, there exists a scheduling policy under which the total expected queue size 
is upper bounded by $c n^{7/3} \log n$ at all times; i.e., for any $\tau \in \NN$, 
\begin{equation}\label{eq:main2}
\EE\left[\sum_{i,j=1}^n Q_{ij}(\tau)\right] \leq c n^{\frac{7}{3}} \log n.
\end{equation}
\end{corollary}
We remark that our upper bound is still a multiplicative factor of $O(n^{1/3} \log n)$ away
from the conjecture scaling $\Theta(n^2)$ in \cite{STZopen}, 
when $1-\rho = \Theta(1/n)$. It is open whether
this gap can be further reduced. 
Let us also remark that we will only prove Theorem \ref{theorem:main} 
for sufficiently large $n$ and $\rho$ that is sufficiently close to $1$. 
For smaller $n$, the theorem remains valid by considering a stabilizing policy
such as Max-Weight, and by re-choosing the constant $c$ to be large enough. 
Similarly, for $\rho$ that is bounded away from $1$ by some fixed amount $\delta >0$, the theorem remains 
valid by considering, for example, the batching policy in \cite{NMC07}, 
and by re-choosing $c$ to be large enough.

\subsection{Policy overview}\label{ssec:policy_overview}
We now proceed to describe the high-level ideas underpinning our proposed policy;
the formal policy description is deferred to Section \ref{sec:policy}.
The policy proposed in this paper, like the one in \cite{STZ16}, is of a batching type. 
In the standard batching policy (see e.g., \cite{NMC07}), 
time is divided into consecutive intervals (batches) of equal lengths, 
and packets that arrive during a batch are served only in subsequent batches. 
%In other words, the time that service commences for arrivals of a given batch is delayed 
%from the beginning of that batch by $b$ time slots.
To guarantee stability, all packets that arrive during one batch 
should be successfully served during the very next batch, with high probability. 
By choosing the smallest batch length that guarantees stability, 
which turns out to be $O\left(\log n/(1-\rho)^2\right)$, 
the policy in \cite{NMC07} gives an upper bound of $O\left(n\log n/(1-\rho)^2\right)$ 
on the expected total queue size. 
See Figure \ref{sfig:standard} for an illustration. 

Different from the standard batching policy in \cite{NMC07}, 
our policy, like the one in \cite{STZ16}, starts serving packets from a given batch much earlier, 
before the arrivals of the entire batch. When $\Omega(n^{-1.5}) \leq 1-\rho \leq O(1/n)$, 
our policy starts serving packets even earlier than the one in \cite{STZ16}, 
resulting in improved queue-size scaling. 
More specifically, our policy starts serving packets from a given batch of length $b = O\left((1-\rho)^{-2} \log f\right)$ 
(recall the definition of $f$ in Eq. \eqref{eq:f})
just after the first $d = O\left((1-\rho)^{-4/3}\log f\right)$ time slots of that batch, 
and is still able to finish serving all packets of this batch in $b$ time slots, with high probability. 
To achieve this, our policy divides a batch into further subintervals of lengths $I_0, I_1, \cdots, I_{\ell}$, 
with $I_0 = d$. Then, the $u$th subinterval is used to serve arrivals from 
the $(u-1)$st subinterval, for $u = 1, 2, \cdots, \ell$. 
See Figure \ref{sfig:our} for an illustration.

\begin{figure}[ht]
  \begin{subfigure}[b]{0.4\textwidth}
    \includegraphics[width=\textwidth]{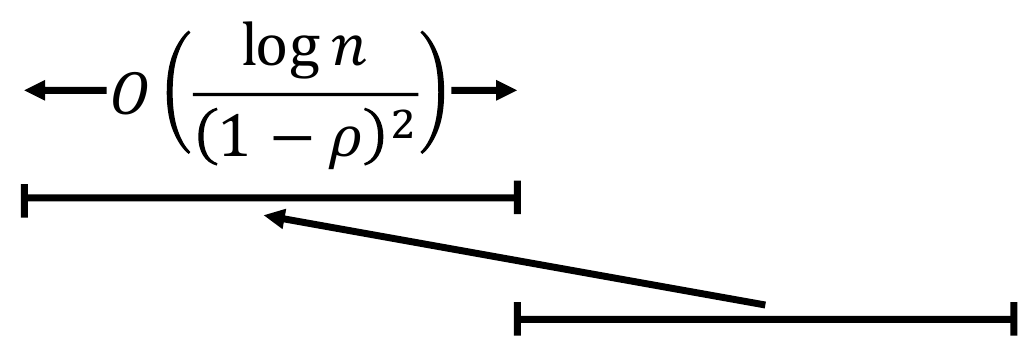}
    \caption{Illustration of a standard batching policy.}
    \label{sfig:standard}
  \end{subfigure}
  \hfill
  \begin{subfigure}[b]{0.45\textwidth}
    \includegraphics[width=\textwidth]{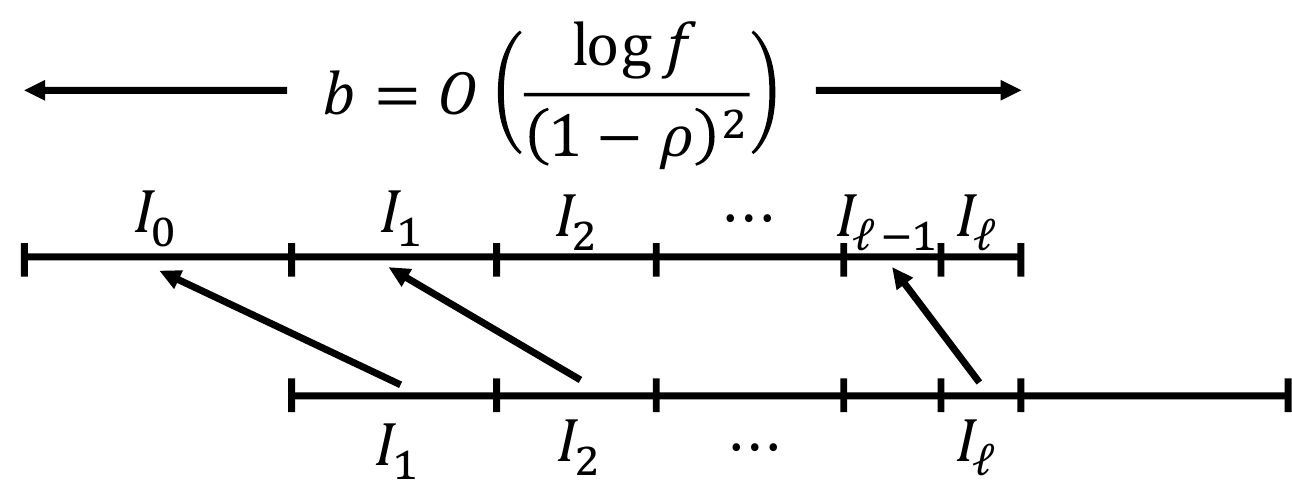}
    \caption{Illustration of our policy.}
    \label{sfig:our}
  \end{subfigure}
\caption{Comparison between a standard batching policy and our policy.}
\label{fig:comparison}
\end{figure}

To guarantee stability, $I_u$'s are carefully chosen so that the subintervals are efficiently utilized. 
More specifically, note that the maximum number of packets that can be served in a period of $I_u$ slots is $nI_u$, 
which can be achieved if and only if a full matching is used in each slot, and no offered service is wasted. 
Then, the following is required of our policy: 
\begin{itemize}
\item For each $u = 1, 2, \cdots, \ell$, $I_u$ is chosen to be largest possible, so that
the $u$th subinterval can serve $nI_u$ packets from the $(u-1)$st subinterval, with high probability. 
\end{itemize}

It is by no means obvious that such a requirement can be met. 
To illustrate the inherent difficulty, consider a $3\times 3$ switch, whose arrivals from 
a subinterval constitutes the following matrix: 
 \begin{equation}\label{ex:low_env_counter}
   \bq = 
  \left[ {\begin{array}{ccc}
   0 & a & 0 \\
   a & 0 & a \\
   0 & a & 0 
  \end{array} } \right],
\end{equation}
where $a$ is some positive integer. 
Then, it is easy to see that for these arrivals, it is impossible to serve $3$ packets in one time slot, 
let alone serving $3t$ packets in any $t$ time slots. 

The fact that the aforementioned requirement can be met has to do with the statistical regularity 
that originates from the stochastic arrivals with uniform rates. In particular, 
using the equivalence between queue matrices and bipartite multigraphs described in the Introduction, 
the requirement that $nI_u$ packets can be served from arrivals of $I_{u-1}$ time slots 
 translates into the existence of an $I_u$-factor of an \ER random bipartite multigraph,  
where the number of edges between a pair of left and right nodes is independently and binomially distributed 
with parameters $I_{u-1}$ and $\rho/n$. Our characterization of the largest $k$-factors in
\ER random bipartite multigraph implies that such a requirement can be met with probability
$1- f^{-16}$, provided that
$$
I_u \le \rho I_{u-1}  - \sqrt{ 304 \rho I_{u-1}  \log f }. 
$$
Hence, we choose $d=I_0 >\cdots >I_\ell = \Theta(d)$ for $\ell=\Theta\left(\sqrt{d/\log f } \right)$.  
Since $\sum_{u=0}^\ell I_u=b$, we get that $d=\Theta\left( b^{2/3} \log^{1/3} f \right)$,
which reduces to $d=\Theta\left( (1-\rho)^{-4/3} \log f \right)$ when $b=\Theta\left( (1-\rho)^{-2} \log f \right)$.
Thus, the expected number of total arrivals in the first subinterval  is 
$n \rho d =\Theta\left( n (1-\rho)^{-4/3} \log f \right)$, which will be shown to 
dictate the order of magnitude of the expected total queue size in Section \ref{sec:analysis}, 
establishing the upper bound in Theorem \ref{theorem:main}.
%be of the same order as 
%our upper bound  in Theorem
%\ref{theorem:main}. 

\section{Preliminaries}\label{sec:prelim} 
Here we provide some preliminary facts that will be useful for our analysis later on. 
The same facts were used in \cite{STZ16} and we state  them here for completeness. 
\paragraph{Concentration Inequalities.} We will use the following concentration bounds 
on the tail probabilities of binomial random variables (adapted from Theorem 2.4 in \cite{FCgraph}). 
\begin{theorem}\label{theorem:bin_tail}
Let $X$ be a binomial random variable with parameters $m$ and $p$, 
so that $\EE[X] = mp$ and $Var(X) = mp(1-p)$. Then, for any $x > 0$, we have 
\begin{align}
\text{(Lower tail)} \quad \quad \PP(X \leq \EE[X] - x) & \leq \exp \left\{-\frac{x^2}{2\EE[X]}\right\}; \label{eq:bin_conc} \\
\text{(Upper tail)} \quad \quad \PP(X \geq \EE[X] + x) & \leq \exp \left\{-\frac{x^2}{2(\EE[X]+x/3)}\right\}. 
\label{eq:bin_conc1}
\end{align}
\end{theorem}
\paragraph{Kingman Bound for Discrete-Time $G/G/1$ Queue.} 
Consider the following discrete-time $G/G/1$ queueing system. 
%\jx{Do we used Kingman Bound in the paper?}
For each $\tau \in \NN$, $X(\tau)$ denotes the number of packets that arrive 
during time slot $\tau$, $Y(\tau)$ denotes the number of packets that can be served 
during time slot $\tau$, and $Z(\tau)$ denotes the queue size 
at the beginning of time slot $\tau$. Suppose that $X(\tau)$ are i.i.d across time, 
so are $Y(\tau)$, and $\left\{X(\tau): \tau \in \NN\right\}$ 
and $\left\{Y(\tau): \tau \in \NN\right\}$ are independent. 
The dynamics of the system is given by
\begin{equation}\label{eq:gg1_dyn}
Z(\tau+1) = \left(Z(\tau) + X(\tau) - Y(\tau)\right)^+.
\end{equation}
The timing convention in Eq. \eqref{eq:gg1_dyn} is different from 
that in Eq. \eqref{eq:dynamics}. In Eq. \eqref{eq:gg1_dyn}, arrivals take place before the services 
in each time slot. This timing convention is used for analyzing the so-called 
{\em backlogged} packets defined in Section \ref{sec:policy} (cf. Eq. \eqref{eq:backlog_dyn}).

%\jx{Here, it seems that we assume the arrivals happen before departures, while
%for the queue update in the switch, we assume the arrivals happen after departures.}
Let $\lambda = \EE[X(\tau)]$, $m_{2x} = \EE\left[X^2(\tau)\right]$, 
$\mu = \EE[Y(\tau)]$ and $m_{2y} = \EE\left[Y^2(\tau)\right]$. 
The following theorem is Theorem 4.2 from \cite{STZ16}, which is a minor adaptation of Theorem 3.4.2 of \cite{SYbook2014}.
\begin{theorem}[Discrete-time Kingman bound]\label{theorem:discrete_kingman}
Consider the aforementioned discrete-time $G/G/1$ system 
with $Z(1) = 0$ and $\lambda < \mu$. Then, 
\begin{equation}\label{eq:kingman}
\EE[Z(\tau)] \leq \frac{m_{2x} + m_{2y} - 2\lambda \mu}{2(\mu-\lambda)}, \quad \text{for all } \tau \in \NN. 
\end{equation}
\end{theorem}

\paragraph{Optimal Clearing Policy.} A key component in any batching-type policy is 
to finish serving all packets that arrive during a batch as quickly as possible. 
Thus, it is important to understand the {\em minimum clearance time} of a queue matrix, 
i.e., the minimum number of slots required to finish serving all packets of a queue matrix using feasible schedules, 
assuming no new arrivals. 
The following theorem (also Theorem 4.3 of \cite{STZ16}) 
%\jx{I do not find a proof of Theorem 4.3 in \cite{}} 
provides a precise characterization of the minimum clearance time.
\begin{theorem}\label{theorem:up_env}
Let $\bq = (q_{ij})_{i,j=1}^n$ be an $n\times n$ queue matrix. Let
\begin{equation}\label{eq:row_col_sum}
r_i = \sum_{j'=1} q_{ij'} \quad \text{and} \quad c_j = \sum_{i'=1}^n q_{i'j}
\end{equation}
be the $i$th row sum and $j$th column sum of $\bq$, respectively. 
Let
\begin{equation}
\gamma = \max \left\{ \max_i r_i, \;  \max_j c_j \right\}.
\end{equation}
Then, $\gamma$ is precisely the minimum clearance time of the queue matrix $\bq$.
\end{theorem}
Note that since in each time slot, at most one packet can be cleared from each row or column, 
each $r_i$ and $c_j$ decreases by at most $1$. Thus, $\gamma$ is clearly a lower bound 
on the minimum clearance time. Theorem \ref{theorem:up_env} states that there exists 
an {\em optimal clearing policy} that finishes serving all packets of $\bq$ in precisely $\gamma$ slots. 

\section{Lower Envelopes of Queue Matrices}\label{sec:envelope}
As mentioned in the Introduction (Section \ref{sec:intro}) and the policy overview (Section \ref{sec:main}), 
of central importance to our policy is the ability to serve arrivals of a subinterval 
using full matchings, without any wasted offered services. This motivates us to define 
the {\em lower envelopes} of a queue matrix. 

\begin{definition}\label{def:envelope}
Let $\bq = (q_{ij})_{i,j=1}^n$ be an $n\times n$ queue matrix. 
A matrix $\bg = (g_{ij})_{i,j=1}^n$ is a $\beta$-{\em lower envelope} of $\bq$ if 
(a) for all $i$ and $j$, $g_{ij} \in \ZZ_+$ and $g_{ij} \leq q_{ij}$; and (b)
\begin{equation}\label{eq:low_env}
\beta = \sum_{i'=1}^n g_{i'j} = \sum_{j'=1}^n g_{ij'} \quad \text{for all } i, j \in [n].
\end{equation}
%A matrix $\bh = (h_{ij})_{i,j=1}^n$ is a $\gamma$-{\em upper envelope} of $\bq$ if 
%(a) for all $i$ and $j$, $h_{ij} \in \ZZ^+$ and $h_{ij} \geq q_{ij}$; and (b)
%\begin{equation}\label{eq:up_env}
%\gamma = \sum_{i'=1}^n h_{i'j} = \sum_{j'=1}^n h_{ij'} \quad \text{for all } i, j \in [n].
%\end{equation}
\end{definition}
\begin{remark}\label{rmk:factor_decomp}
Let us note that if $\bg$ is a $\beta$-lower envelope of queue matrix $\bq$, 
then by Theorem \ref{theorem:up_env}, packets of $\bg$, which also belong to $\bq$, 
can be cleared optimally using $\beta$ full matchings, without any wasted service. 
We call any such $\beta$ full matchings the {\em schedules prescribed by} the $\beta$-lower envelope $\bg$ of $\bq$.
\end{remark}

Recall that a queue matrix $\bq$ can be viewed equivalently as the biadjacency matrix of a bipartite multigraph $G$
with left vertex set $[n]$ and right vertex set $[n]$, where 
left vertex $i$ is connecting to right vertex $j$ with $q_{ij}$ multiple edges if $q_{ij}\ge 1$.
In this vein, a $\beta$-lower envelope $\bg$ of $\bq$ can be equivalently viewed as a 
$\beta$-factor of $G$, i.e., a spanning $\beta$-regular subgraph of $G$.
Note that a $1$-factor of $G$ is simply a perfect matching of $G$.  Due to this equivalence, 
we will use the terminologies ``a $\beta$-lower envelope of a queue matrix'' and
``a $\beta$-factor of a bipartite multigraph'' interchangeably, whichever is more convenient. 

%\begin{example}\label{ex:low_env_counter} 
%The following $3\times 3$ queue matrix $\bQ$ does not have any $\beta$-lower envelope with $\beta>0$ for any $m$:
% \begin{equation}\label{ex:low_env_counter}
%   \bq = 
%  \left[ {\begin{array}{ccc}
%   0 & m & 0 \\
%   m & 0 & m \\
%   0 & m & 0 
%  \end{array} } \right].
%\end{equation}
%\end{example}

The next proposition provides a necessary and sufficient condition for the existence of 
a $\beta$-lower envelope of any queue matrix, and is a simple adaptation of the Gale-Ryser theorem~\cite{gale1957theorem,ryser1957combinatorial}. 
Let $\cR, \cC \subseteq [n]$ be non-empty subsets of $[n]$. 
For an $n\times n$ queue matrix $\bq = (q_{ij})_{i,j=1}^n$, 
we use $\bq_{\cR, \cC}$ to denote the submatrix $(q_{ij})_{i \in \cR, j \in \cC}$. 

\begin{proposition}\label{prop:exist_lowenv}
Consider an $n\times n$ queue matrix $\bq = (q_{ij})_{i,j=1}^n$. 
There exists a $\beta$-lower envelope of $\bq$ if and only if for any $k \times \ell$ 
submatrix $\bq_{\cR, \cC} = (q_{ij})_{i \in \cR, j \in \cC}$ of $\bq$, 
where $\cR, \cC \subseteq [n]$, $|\cR| = k$ and $|\cC| = \ell$, 
we have 
\begin{equation}\label{eq:exist_lowenv}
\beta (k+\ell-n) \leq \sum_{i\in \cR, j \in \cC} q_{ij}.
\end{equation}
\end{proposition}

%[MAY MOVE THIS PROOF TO APPENDIX.]
\begin{proof} 
The proposition follows as a simple corollary of Feasibility Theorem in~\cite{gale1957theorem}. For the sake of completeness and ease of reference, we record here a direct proof based on the max-flow min-cut theorem~\cite{ford1956maximal}.

Consider a network $\bar{G}$ with a source node $s$, a sink node $t$, $n$ left vertices corresponding to the rows of $\bq$, and $n$ right vertices corresponding to the columns of $\bq$. 
The source node $s$ is connecting to every left vertex $i \in [n]$ with a directed edge $(s,i)$ of capacity $\beta$. 
The sink node $t$ is connecting to every right vertex $j \in [n]$ with a directed edge $(j, t)$ of capacity $\beta$. Moreover, every left vertex $i$ is connected to  every right vertex 
$j$ with a directed edge $(i,j)$ of capacity $q_{ij}$. 

We first show \eqref{eq:exist_lowenv} is necessary. Suppose there exists a $\beta$-lower envelope of $\bq$, denoted by $\bg$. 
Define $f(s,i)=\beta$ for every left vertex $i \in [n]$,
$f(j,t)=\beta$ for every right vertex $j \in [n]$, and $f(i,j)=g_{ij}$ for every left vertex $i \in [n]$ and every right vertex $j \in [n]$.
Then $f$ is a network flow from $s$ to $t$ on $\bar{G}$. Hence, the max flow from $s$ to $t$ on network $\bar{G}$ is at least $\beta n$. 
For any given set $\cR \subset [n]$ of left vertices and set $\cC \subset [n]$ of right vertices, 
consider a cut $S=\{s\}\cup \cR \cup ([n]-\cC)$ and $S^c=\{t\} \cup ([n]-\cR) \cup \cC$. The capacity of cut $(S,S^c)$ is 
$$
c(S, S^c) \triangleq \beta \left( n - |\cR| \right) + \beta \left( n - |\cC| \right) + \sum_{i\in \cR, j \in \cC} q_{ij}.
$$
According to the max-flow min-cut theorem, the value of the max flow equals to the minimum cut capacity. Thus 
$$
\beta n \le \beta \left( n - |\cR| \right) + \beta \left( n - |\cC| \right) + \sum_{i\in \cR, j \in \cC} q_{ij},
$$
which is equivalent to  \eqref{eq:exist_lowenv}. 

We next show \eqref{eq:exist_lowenv} is sufficient. Suppose  \eqref{eq:exist_lowenv} holds.  
Let $(S, S^c)$ be any cut such that $s \in S$ and $t \in S^c$. 
Let $A$ denote the set of left vertices and $B$ denote the set of right vertices. 
The capacity of cut $(S,S^c)$ satisfies
\begin{align*}
c(S, S^c) & = c(S-s, S^c - t)  + c(s, S^c-t) +  c(S-s, t) \\
& = c\left((S-s) \cap A, (S^c-t) \cap B \right) + \\
& \quad +  \beta \left| (S^c-t) \cap A \right| + \beta \left| (S-s) \cap B \right|,
\end{align*}
where we write $S-s=S\setminus \{s\}$ for ease of notation. 
Let $\cR= (S-s) \cap A$ and $\cC=(S^c-t)\cap B$. Then 
$
\left| (S^c-t) \cap A \right| = n - |\cR| 
$
and $
\left| (S-s) \cap B \right| = n-|\cC|$.
It follows that 
\begin{align*}
c(S, S^c) & = c\left(\cR, \cC \right) + \beta \left( n - |\cR| \right) + \beta \left( n - |\cC| \right) \\
&= \sum_{i\in \cR, j \in \cC} q_{ij} + \beta \left( n - |\cR| \right) + \beta \left( n - |\cC| \right)  \ge \beta n,
\end{align*}
where the last inequality holds due to  \eqref{eq:exist_lowenv}. 
Moreover, if $S=\{s\}$, then the capacity of the cut $(S, S^c)$ is equal to $\beta n$.
Therefore,  the minimum cut capacity of 
network $\bar{G}$ with respect to $s$ and $t$ is 
$\beta n$.

Now, from the max-flow min-cut theorem, there is a flow $f$ from $s$ to $t$ on $\bar{G}$ such that 
the value of the flow is $\beta n$. In particular, for every left vertex $i \in [n] $ and right vertex $j \in [n]$,
$$
\beta  = f(s, i)  = \sum_{j' \in [n]} f(i, j' ) = \sum_{i' \in [n]} f(i', j) = f(j,t).
$$
 Moreover, since the edge capacity is integer-valued,
this flow $f$ can be further chosen to be integer-valued. 
Define an $n\times n$ matrix
$\bg$ such that $g_{ij}= f(i,j)$ for every left vertex $i$ and every right vertex $j$. 
Then $g_{ij} \in \ZZ_+$ and 
$$
\beta = \sum_{j' \in [n]} g(i, j' ) = \sum_{i' \in [n]} g(i', j).
$$
Therefore, $\bg$ is a $\beta$-lower envelope of $\bq$, completing the proof. 
\end{proof}

The main result of this section is the following tight characterization of 
the existence of a $\beta$-lower envelope for a random queue matrix. 
\begin{theorem}\label{theorem:rand_low_env}
Let $\bQ = (Q_{ij})_{i,j=1}^n$ be an $n\times n$ random queue matrix 
with i.i.d entries that are binomially distributed with parameters $m$ and $p$. 
%such that $np > 1/2$ and $m \geq n$. 
Let $f$ be an additional, positive parameter such that $f \geq n$.
Suppose
% and 
%let $m$ and $f$ satisfy
\begin{equation}\label{eq:fm_cond}
pmn \geq 152 \log f.
\end{equation}
Then, with probability $1 - 1/f^{16}$, $\bQ$ has a $\beta_0$-lower envelope with 
\begin{equation}\label{eq:beta0}
\beta_0 = \lfloor pmn - \sqrt{304 pmn \log f} \rfloor.
\end{equation}
\end{theorem}
\begin{remark}
Note that for any  $\beta$-lower envelope $\bQ$, it must satisfy that
$$
\beta \le \min \left\{ \min_i r_i, \; \min_j c_j \right\},
$$
where $r_i$ and $c_j$ as given in \eqref{eq:row_col_sum} 
are the $i$th row sum and $j$th column sum of $\bq$, respectively.
Since both $r_i$ and $c_j$ are binomially distributed with parameters $mn$ and $p$, 
one can show that if $mnp \ge \delta_1 \log n$ for a sufficiently large constant 
$\delta_1>0$, then there exist a constant $\delta_2>0$ such that with high probability, %probability converging to $1$, 
$$
 \min_i r_i \le pmn - \delta_2 \sqrt{pmn \log n}
$$
and hence $\beta \le pmn - \delta_2 \sqrt{pmn \log n}.$ Therefore, the lower bound to $\beta$ given in
\eqref{eq:beta0} with $f=n$ is tight up to a constant factor. 
\end{remark}
\begin{remark}
Theorem \ref{theorem:rand_low_env} holds for the special case where $m=1$. In this case,
$\bQ$ can be viewed as the biadjacency matrix of an \ER bipartite simple graph $G$
with edge probability $p$. Therefore, our result
implies that if $np=\omega(\log n)$, then with probability converging to $1$, 
$G$ has $\beta$-factor with $\beta\sim np$.

In comparison, previous work~\cite{csaba2007regular} shows that any simple bipartite graph
with $n$ left (right) vertices and minimum degree $\delta n$ where $\delta \geq 1/2$, 
has a $\beta$-factor with $\beta =\lfloor n \frac{\delta+\sqrt{2\delta-1}}{2}\rfloor$.
Since with high probability 
the \ER bipartite graph $G$ has minimum degree $\sim np$ if $p \ge 1/2$,
the result in~\cite{csaba2007regular} implies that if $p \ge 1/2$, then 
$G$ has $\beta$-factor with $\beta\sim n \frac{p+\sqrt{2p-1}}{2}$ with high probability.

Note that $\frac{p+\sqrt{2p-1}}{2} \le p$ with equality if and only $p \in \{0,1\}.$
Therefore, our result improves over  the previous result~\cite{csaba2007regular} in the context of \ER bipartite random graphs. As can be seen in the proof, we achieve this improvement by exploiting certain nice random structures in $\bQ.$
\end{remark}

%[MAY WISH TO MOVE PROOF OF PROPOSITION TO APPENDIX FOR SIGMETRICS; WE ARE ALREADY HALFWAY INTO THE PAPER AT THIS POINT. WE'LL SEE ABOUT SPACE THEN DECIDE.]

%[RELATE TO RESULTS FROM BELA CSABA; SAY RESULT ALSO HOLDS FOR SIMPLE GRAPHS SUCH AS ERDOS-RENYI BIPARTITE.]
 
\begin{proof}
By Proposition \ref{prop:exist_lowenv}, it suffices to show that with high probability, the inequalities 
\begin{equation}\label{eq:exist_lowenv0}
\beta (|\cR|+|\cC|-n) \leq \sum_{i\in \cR, j \in \cC} Q_{ij}
\end{equation}
hold for any non-empty subsets $\cR$ and $\cC$ of $[n]$, where 
\begin{equation}\label{eq:rand_low_env}
\beta = pmn - \sqrt{304 pmn \log f}.
\end{equation} 
To this end, fix some $k, \ell \in [n]$, and suppose that $k \geq \ell$. 
WLOG, also suppose that $k+\ell > n$, since the inequality \eqref{eq:exist_lowenv0} is satisfied deterministically 
for subsets $\cR$ and $\cC$ of $[n]$ if $|\cR| + |\cC| \leq n$. 
There are ${n \choose k} = {n \choose n-k} \leq n^{n-k}$ subsets of $[n]$ with size $k$, 
and ${n \choose \ell} \leq n^{\ell}$ subsets of $[n]$ with size $\ell$, 
so the number of $k \times \ell$ submatrices of $\bQ$ is no larger than $n^{n-k}\cdot n^{\ell} = n^{n-k+\ell}$. 
Furthermore, for any fixed $k\times \ell$ submatrix, say $\bQ_{\cR, \cC}$, of $\bQ$, 
since the entries of $\bQ_{\cR, \cC}$ are i.i.d binomial random variables with parameters $m$ and $p$, 
the sum $\sum_{i\in \cR, j \in \cC} Q_{ij}$ is binomially distributed with parameters $k\ell m$ and $p$. 
Thus, by Theorem \ref{theorem:bin_tail}, for fixed subsets $\cR$ and $\cC$ of $[n]$ with $|\cR|=k$ and $|\cC|=\ell$, 
we have
\begin{align}
&~\PP\left(\sum_{i\in \cR, j \in \cC} Q_{ij} \leq k\ell mp -\sqrt{38(n-k+\ell) k\ell mp \log f}\right) \nonumber \\
\leq &~\exp\left(-\frac{38(n-k+\ell) k\ell mp \log f}{2k\ell mp}\right) = f^{-19(n-k+\ell)}. \label{eq:bin1}
\end{align}
Let $\udl{\bQ}(k, \ell)$ denote the minimum value over the sums of entries of all $k \times \ell$ submatrices of $\bQ$, 
i.e., 
\begin{equation}\label{eq:Qkl}
\udl{\bQ}(k, \ell) = \min_{\cR, \cC \subseteq [n]: |\cR|=k, |\cC|=\ell} \sum_{i \in \cR, j \in \cC} Q_{ij}.
\end{equation}
Fix subsets $\cR$ and $\cC$ of $[n]$ with sizes $k$ and $\ell$, respectively. Then, by union bound, we have
\begin{align}
&~\PP\left(\udl{\bQ}(k, \ell) \leq k\ell mp - \sqrt{38(n-k+\ell) k\ell mp \log f} \right) \nonumber \\
\leq &~n^{n-k+\ell} \PP\left(\sum_{i\in \cR, j \in \cC} Q_{ij} \leq k\ell mp - \sqrt{38(n-k+\ell) k\ell mp \log f}\right) \nonumber \\
\leq &~n^{n-k+\ell} \cdot f^{-19(n-k+\ell)} \leq f^{-18(n-k+\ell)} \leq f^{-18}, \label{eq:lowenv_tail1}
\end{align}
where, for the second last inequality, we used the fact that $f \geq n$, and for the last inequality, 
we used the fact that $n-k+\ell \geq \ell \geq 1$.
Thus, by union bound again, 
\begin{align}
&~\PP~\Bigg(\udl{\bQ}(k, \ell) \leq k\ell mp - \sqrt{38(n-k+\ell) k\ell mp \log f} \nonumber \\
&\quad \text{ for some } (k, \ell) \text{ with } k \geq \ell, k+\ell >n\Bigg) \nonumber \\
\leq &~ \frac{1}{2}n(n+1) f^{-18}, \label{eq:lowenv_tail2}
\end{align}
since there are at most $n(n+1)/2$ pairs of $(k, \ell) \in [n]\times [n]$ with $k \geq \ell$. 
Similarly, we can get 
\begin{align}
&~\PP~\Bigg(\udl{\bQ}(k, \ell) \leq k\ell mp - \sqrt{38(n-\ell+k) k\ell mp \log f} \nonumber \\
&\quad \text{ for some } (k, \ell) \text{ with } k < \ell, k+\ell >n\Bigg) \nonumber \\
\leq &~ \frac{1}{2}n(n-1) f^{-18}. \label{eq:lowenv_tail3}
\end{align}
Thus, by Ineq. \eqref{eq:lowenv_tail2} and \eqref{eq:lowenv_tail3}, and by union bound, we have 
\begin{align}
&~\PP~\Bigg(\udl{\bQ}(k, \ell) \leq k\ell mp - \sqrt{38(n-|k-\ell|) k\ell mp \log f} \nonumber \\
&\quad \text{ for some } (k, \ell) \text{ with } k+\ell >n\Bigg) \nonumber \\
\leq &~ n^2 f^{-18} \leq f^{-16}, \label{eq:lowenv_tail4}
\end{align}
where for the last inequality, we used the fact that $f\geq n$.
We now claim that for all $k, \ell \in [n]$ with $k+\ell > n$,
\begin{equation}\label{eq:exist_lowenv1}
\beta (k+\ell-n) \leq k\ell mp - \sqrt{38(n - |k-\ell|) k\ell mp \log f},
\end{equation}
where $\beta$ is given by Eq. \eqref{eq:rand_low_env}. 

{\em Proof of the claim.} We first prove the claim for the case where $k \geq \ell$ and $k+\ell > n$. 
In this case, $k > n/2$, and $n-|k-\ell| = n-k+\ell < 2 \ell$. To prove the claim in this case, 
it suffices to show that for any $k \in [n]$ such that $k > n/2$, 
\begin{equation}\label{eq:exist_lowenv2}
\beta \left(2-\frac{n}{k}\right) \leq kmp - \sqrt{76 kmp \log f}.
\end{equation}
To see that Ineq. \eqref{eq:exist_lowenv2} implies Ineq. \eqref{eq:exist_lowenv1}, 
note that because $k \geq \ell$, we have
\begin{equation}
2-\frac{n}{k} = 1 - \frac{n-k}{k} \geq 1 - \frac{n-k}{\ell} = \frac{k+\ell-n}{\ell}. 
\end{equation}
Thus, Ineq. \eqref{eq:exist_lowenv2} implies that 
\begin{equation}
\beta \cdot \frac{k+\ell-n}{\ell} \leq kmp - \sqrt{76 kmp \log f}, 
\end{equation}
and then 
\begin{align}
\beta (k+\ell-n) \leq &~k\ell mp - \ell \sqrt{76 kmp \log f} \nonumber \\
= &~k\ell mp - \sqrt{38\cdot 2\ell\cdot k \ell mp \log f} \nonumber \\
\leq &~k\ell mp - \sqrt{38(n-k+\ell) k \ell mp \log f}, 
\end{align}
recovering Ineq. \eqref{eq:exist_lowenv1}. 
Here, the last inequality follows from the assumption that $k+\ell > n$. 

Next, we prove Ineq. \eqref{eq:exist_lowenv2}. 
Write $x = k/n$, and $\rho = np$. Then, using the expression \eqref{eq:rand_low_env} for $\beta$, 
we can re-write Ineq. \eqref{eq:exist_lowenv2} as 
\begin{equation}\label{eq:exist_lowenv3}
\left(2-\frac{1}{x}\right) \left(\rho m - \sqrt{304\rho m \log f}\right) \leq 
x \rho m - \sqrt{x}\cdot \sqrt{76 \rho m \log f}.
\end{equation}
Re-arranging Ineq. \eqref{eq:exist_lowenv3}, we obtain the equivalent expression 
\begin{equation}\label{eq:exist_lowenv4}
\left(x+\frac{1}{x}-2\right) \rho m\geq \left(\sqrt{x} + \frac{2}{x} - 4\right) \sqrt{76 \rho m \log f},
\end{equation}
Note that $x \in (1/2, 1]$ because $n/2 < k \leq n$.
We will establish next
\begin{align}
x + \frac{1}{x} - 2 \ge \frac{1}{\sqrt{2}} \left(\sqrt{x} + \frac{2}{x} - 4\right), \quad \forall x \in [1/2,1] \label{eq:desired_ineq},
\end{align}
which together with the condition \eqref{eq:fm_cond} where $\rho m \geq 2\cdot 76 \log f$ immediately implies \eqref{eq:exist_lowenv4}.

Note that  \eqref{eq:desired_ineq} clearly holds for $x=1$. 
For any $x \in [1/2, 1)$, $x+\frac{1}{x}-2>0$
and moreover 
$$
\frac{\sqrt{x} + \frac{2}{x} - 4}{x + \frac{1}{x} - 2} = 2 - \frac{2x - \sqrt{x} }{x + \frac{1}{x} - 2}.
$$
Note that $2x-\sqrt{x}=\sqrt{x} (2\sqrt{x}-1)$ is increasing in $x \in [1/2,1)$,
while $x+1/x$ is decreasing in $x \in [1/2,1)$. 
Hence,
$$
\sup_{x \in [1/2,1)} \frac{\sqrt{x} + \frac{2}{x} - 4}{x + \frac{1}{x} - 2}
= \frac{\sqrt{x} + \frac{2}{x} - 4}{x + \frac{1}{x} - 2} \bigg|_{x=1/2} = \sqrt{2},
$$
completing the proof of \eqref{eq:desired_ineq}.

 To summarize, we have established Ineq. \eqref{eq:exist_lowenv4} that is 
 equivalent to Ineq. \eqref{eq:exist_lowenv2}, which implies Ineq. \eqref{eq:exist_lowenv1}, 
 for the case where $k \geq \ell$ and $k+\ell > n$. 
 By symmetry, the proof of Ineq. \eqref{eq:exist_lowenv1} for the case where $k < \ell$ and $k+\ell > n$ 
 follows exactly the same line of argument, with the roles of $k$ and $\ell$ interchanged. 
 Thus, we have established Ineq. \eqref{eq:exist_lowenv1} and hence the claim. $\Box$

To complete the proof of the theorem, note that by Ineq. \eqref{eq:exist_lowenv1} and \eqref{eq:lowenv_tail4}, 
we have that with probability $1-f^{-16}$, 
\begin{equation}
\beta (k+\ell-n) < \udl{\bQ}(k, \ell), \text{ for all } (k, \ell) 
\text{ with } k+\ell > n, 
\end{equation}
in which case the matrix $\bQ$ has a $\beta$-lower envelope, by Proposition \ref{prop:exist_lowenv}. 
This completes the proof of the theorem. 
\end{proof}

\section{Policy Description}\label{sec:policy}
%[ONLY DESCRIBE POLICY FOR SUFFICIENTLY LARGE $n$.]
%
%At a high level, the policy that we employ in this paper is similar to the one in \cite{}. 
%Both policies are of the batching type, and, instead of waiting for the arrivals of an entire batch to take place, 
%both policies wait for just enough arrivals and start serving a batch earlier. 
%A key difference is that the policy of this paper waits for fewer arrivals to start service than 
%the policy in \cite{}, which enables us to obtain sharper scaling bounds for the expected total queue size. 
%
%We now proceed to a detailed description of the scheduling policy. 
This section contains a detailed description of our proposed scheduling policy. 
There are similarities between our policy and the one in \cite{STZ16}, but there are also key differences. 
Thus, the policy description will be complete and self-contained, 
with comparisons made to the policy in \cite{STZ16} throughout. 
Let us also note that it suffices for our policy to be well-defined for sufficiently large $n$ and $\rho$ 
that is sufficiently close to $1$, 
since the main result, Theorem \ref{theorem:main}, will only be proved for sufficiently large $n$ 
and $\rho$ sufficiently close to $1$.

Recall that we define $f = n \vee (1-\rho)^{-1}$ in
\eqref{eq:f} for notational convenience. We introduce three parameters, $b$, $d$ and $s$. 
They define the lengths of certain time intervals, which in turn specify different phases of the policy. 
These parameters are given by 
\begin{align} \label{eq:policy_pm}
b &= c_b (1-\rho)^{-2} \log f; \\
d &= c_d (1-\rho)^{-4/3} \log f; \text{ and } \label{eq:policy_d}\\
s &= \rho b + \sqrt{c_s b\log f}.  
\end{align}
Here, $c_b$, $c_d$ and $c_s$ 
are positive constants that are independent of $n$ and $\rho$, and chosen so that
\begin{equation}\label{eq:const_cond}
c_b - \sqrt{c_s c_b} \geq 1, 
%\quad c_d^2 \geq 640 c_b, 
\quad c_d^{3/2} \geq 76 c_b, \quad c_d \geq c_b, \quad c_s \geq 30.
\end{equation}
We also assume that 
\begin{equation}\label{eq:nrho_cond}
n \geq 4, \quad \rho \geq \frac{1}{2}, \quad \frac{1}{(1-\rho)^{2/3}} \geq \frac{38}{\sqrt{c_d}} \vee \frac{\sqrt{c_d}}{38} \vee c_d. 
\end{equation}
Similar to \cite{STZ16}, we will treat the parameters $b$, $d$ and $s$ as integers, 
since doing so will not affect our order-of-magnitude estimates. 
%We also assume $n \geq 3$, so that $\log f \geq 1$. 
%For the rest of the paper, we suppress the subscript $n$ from parameters $b_n$, $d_n$ and $s_n$, 
%and write $b$, $d$ and $s$, to simplify notation.

Similar to the policy in \cite{STZ16}, we divide time into consecutive intervals of length $b$. 
For $k = 0, 1, 2, \cdots$, the interval consisting of time slots $kb+1, kb+2, \cdots, (k+1)b$ 
is called the {\em $k$th arrival period}, or the {\em $k$th batch}. 
The policy intends to complete serving all $k$th-batch arrivals in the {\em $k$th service period}, 
which is offset from the $k$th batch by a delay of $d$ and consists of time slots 
$kb+d+1, \cdots, (k+1)b+d$. All $k$th-batch packets that are not served during the $k$th service period 
are called {\em backlogged} and will be served in subsequent service periods. 
%\jx{It seems that the backlogged $k$th-batch packets will be served in the the Backlog Clearing phase in
%$k$th service period}
As we will see, the policy is designed in a way so that the number of backlogged packets will be zero with high probability. 

We now describe the policy in more detail, by first considering what may happen  
before the start of the $k$th service period, $k \in \ZZ_+$.
There may be backlogged packets from previous service periods, 
whose number we denote by $B_k$. By convention, $B_0 = 0$. 
In addition, time slots $kb, kb+1, \cdots, kb+d$ 
do not belong to the $k$th service period; arrivals from these $d$ slots do belong to the $k$th arrival period, 
however, and they accumulate without being served until slot $kb+d+1$, the beginning of the $k$th service period. 

Next, during the $k$th service period, scheduling decisions take place in three phases, 
which are described below and illustrated in Figures \ref{sfig:policy} and \ref{sfig:lower_env}. 
Note that the scheduling decisions in the first phase, the {\em lower envelope phase}, differ substantially 
from those in the first phase of the policy in \cite{STZ16}.
\begin{figure}[ht]
  \begin{subfigure}[b]{0.5\textwidth}
    \includegraphics[width=\textwidth]{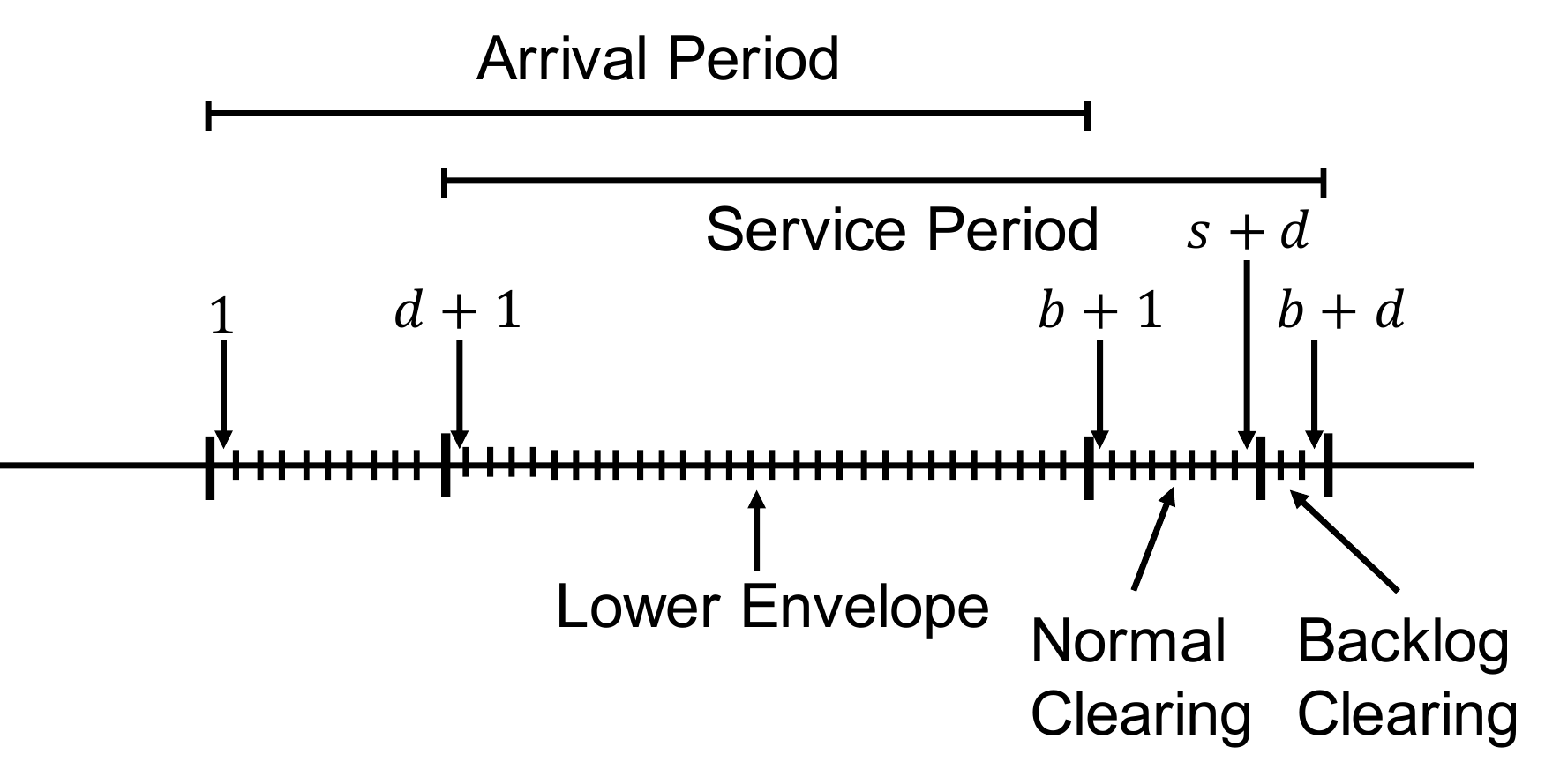}
    \caption{Illustration of a typical arrival period and the phases of a service period. 
    The first slot of the arrival period is numbered as slot $1$, and subsequent slots 
    are numbered consecutively. The same numbering convention is used in Figure \ref{sfig:lower_env} as well.}
    \label{sfig:policy}
  \end{subfigure}
  \hfill
  \begin{subfigure}[b]{0.45\textwidth}
    \includegraphics[width=\textwidth]{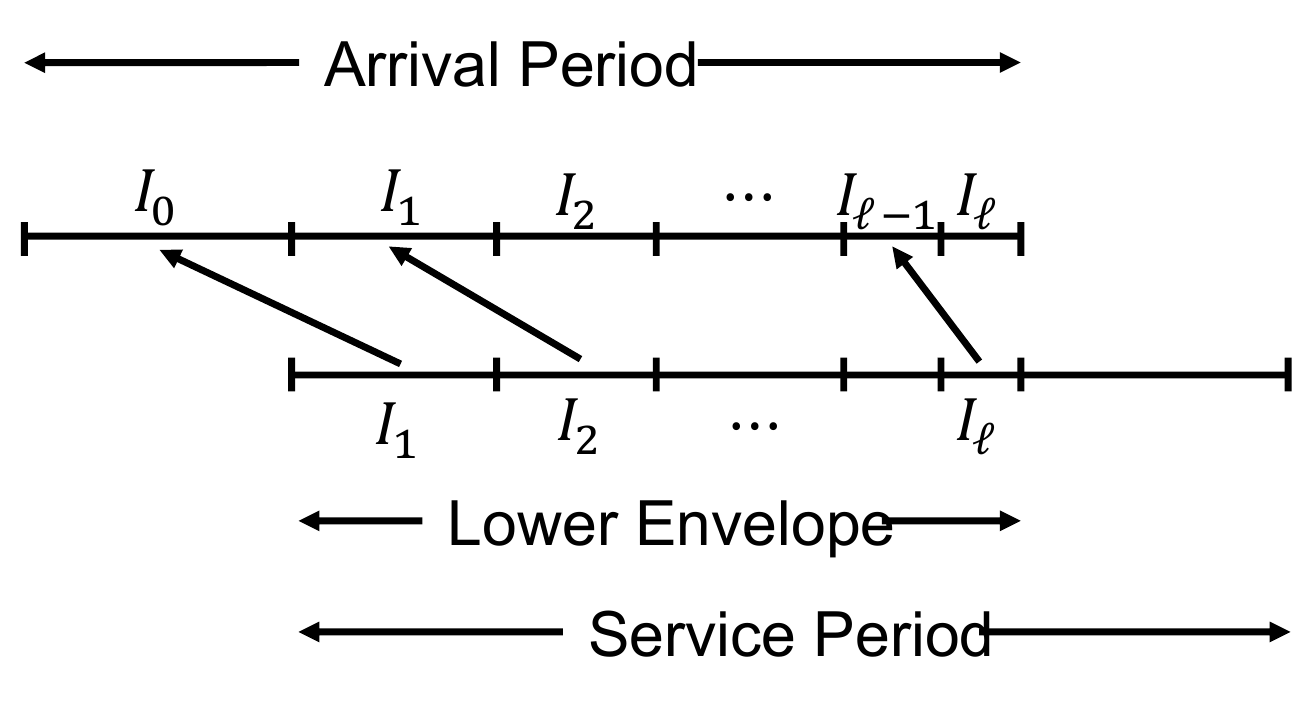}
    \caption{Illustration of the lower envelope phase of the service period. 
    The subinterval $I_0$ consists of time slots $1$ to $d$. The lower envelope phase 
    is divided into consecutive subintervals $I_1, I_2, \cdots, I_{\ell}$, with $I_u$ serving 
    arrivals from $I_{u-1}$, $u \in \{1, 2, \cdots, \ell\}$.}
    \label{sfig:lower_env}
  \end{subfigure}
\caption{Illustration of the different phases of our policy.}
\label{fig:policy}
\end{figure}

\begin{enumerate}
\item The $k$th {\em Lower Envelope Phase} consists of the first $b-d$ time slots of the $k$th service period; 
namely, slots $kb+d+1, \cdots, (k+1)b$. This phase is divided further into consecutive {\em subintervals}, 
which we denote by $I_1, I_2, \cdots, I_{\ell}$, where the value of $\ell$ will be specified shortly. 
Let $I_0$ denote the subinterval consisting of 
slots $kb, kb+1, \cdots, kb+d$. With a slight abuse of notation, we also let $I_u$ denote the length of 
the subinterval $I_u$, $u \in \{0, 1, \cdots, \ell\}$, and they are specified as follows:
\begin{equation}\label{eq:sint}
I_u = \left\{\begin{array}{ll}
d - 19 u\sqrt{d \log f}, & \text{ if } u \in \{0, 1, \cdots, \ell-1\}; \\
b - \left(I_0 + I_1 + \cdots + I_{\ell-1}\right) & \text{ if } u = \ell.
\end{array}\right.
\end{equation}
Here, $\ell$ is defined so that $I_{\ell}$ satisfies the condition 
\begin{equation}\label{eq:ell_cond}
0 \leq I_{\ell} \leq d - 19 \ell \sqrt{d \log f}. 
\end{equation}
In Lemma \ref{lem:ell}, we will show that $\ell$ is well-defined, i.e., there exists a unique $\ell \in \NN$ 
such that condition \eqref{eq:ell_cond} holds, and provide order-of-magnitude estimates for $\ell$. 

We now proceed to complete the policy description for the lower envelope phase, by 
describing the scheduling decisions made in each subinterval $I_1, \cdots, I_{\ell}$. 
Let $u \in [\ell]$. By the beginning of subinterval $I_u$, all packets from subinterval $I_{u-1}$ have already arrived. 
During subinterval $I_u$, the policy only handles arrivals from $I_{u-1}$, 
even though there may be unserved packets from times before $I_{u-1}$. 
More specifically, let $\bA^{(u-1)} = \left(A^{(u-1)}_{ij}\right)_{i,j=1}^n$ denote the matrix of 
the numbers of arrivals to the $n^2$ queues during subinterval $I_{u-1}$. 
Let $\beta$ be maximal such that $\bA^{(u-1)}$ has a $\beta$-lower envelope. 
Then, $\bA^{(u-1)}$ also has a $\min\{\beta, I_u\}$-lower envelope. 
During subinterval $I_u$, the policy uses $\min\{\beta, I_u\}$ schedules 
prescribed by the $\min\{\beta, I_u\}$-lower envelope 
(recall Remark \ref{rmk:factor_decomp} after Definition \ref{def:envelope}, 
and note that all schedules prescribed by the lower envelope are full matchings). 
If $\beta < I_u$, the policy {\em idles} for the remaining $I_u - \beta$ time slots. 
%If $\beta \geq I_u$, there may be arrivals from $I_{u-1}$ that are not served during subinterval $I_u$; 
By the end of the lower envelope phase, there may be $k$th-batch arrivals 
that still have not been served; all these arrivals are handled in the normal clearing phase, 
and possibly the backlog clearing phase and subsequent service periods.

\item The $k$th {\em Normal Clearing Phase} consists of the next $d+s-b$ slots, 
namely slots $(k+1)b+1, \cdots, kb+d+s$. During this phase, our policy makes the same decisions as that in \cite{STZ16}. 
More specifically, the policy does the following. First, during this phase, 
the policy does not serve any backlogged packets, or any arrivals from the $(k+1)$st batch 
(namely, arrivals from the slots $(k+1)b+1, \cdots, kb+d+s$, 
which take place concurrently with the phase). 
Next, by slot $(k+1)b+1$, the start of the normal clearing phase, all packets from the $k$th batch have already arrived. 
Some of these packets have already been served during the lower envelope phase; 
the remaining ones will be served using the optimal clearing policy described in Section \ref{sec:envelope}; 
cf. the discussion after Theorem \ref{theorem:up_env}. 
It is possible that all $k$th-batch packets are successfully served before the end of the normal clearing phase, 
in which case the policy simply idles for the remaining time slots of this phase. 
On the other hand, it is also possible that by the end of the normal clearing phase, 
there are still unserved $k$th-batch packets, whose number we denote by $U_k$. 
These $U_k$ packets are considered backlogged and then added to the backlog $B_k$ from 
previous service periods.

\item The $k$th {\em Backlog Clearing Phase} consists of the last $b-s$ slots, 
namely slots $kb+d+s+1, \cdots, (k+1)b+d$. Again, the scheduling decisions in this phase 
are the same as those in \cite{STZ16}. 
During this phase, the policy only handles the $B_k+U_k$ backlogged packets accumulated thus far, 
and it does not serve any arrivals from the $(k+1)$st batch, 
similar to what happens during the normal clearing phase. 
We serve the $B_k+U_k$ backlogged packets using some arbitrary policy, 
with the only requirement that the policy serves one packet in each time slot whenever 
there are still backlogged packets remaining. 
Let $B_{k+1}$ denote the number of unserved backlogged packets from the backlog clearing phase, 
which is also exactly the number of backlogged packets at the beginning of the $(k+1)$st service period. 
Based on our description, we can see that the numbers of backlogged packets follow the dynamics given by
\begin{equation}\label{eq:backlog_dyn}
B_{k+1} = \left(B_k+U_k - (b-s)\right)^+, \quad k \in \ZZ_+.
\end{equation}
\end{enumerate}
We have now completed the description of the scheduling policy. 
Note that the total length of the three phases of a service period is $(b-d) + (d+s-b) + (b-s) = b$, 
so a service period has the same length as an arrival period. 
To make sure that the policy is well-defined, we need to verify that the following quantities are nonnegative  
and well-defined: $b-d$, the length of the lower envelope phase; 
$I_1, \cdots, I_{\ell}$, the lengths of the subintervals of the lower envelope phase, 
and $\ell$, the number of these subintervals; 
$d+s-b$, the length of the normal clearing phase; 
and $b-s$, the length of the backlog clearing phase. 
This is accomplished in the next few lemmas. 
\begin{lemma}\label{lem:lep}
The length $b-d$ of the lower envelope phase satisfies 
\begin{equation}\label{eq:lep}
b-d \geq \frac{1}{2} b \geq 1. 
\end{equation}
\end{lemma}
\begin{proof} By \eqref{eq:const_cond}, $c_b \geq 30$, 
and then by \eqref{eq:nrho_cond}, $(1-\rho)^{-2/3} \geq 2c_d/c_b$. Thus,  
\begin{align*}
b-d =&~ c_b (1-\rho)^{-2} \log f - c_d (1-\rho)^{-4/3} \log f \\
= &~ \left(c_b (1-\rho)^{-2/3} - c_d\right) (1-\rho)^{-4/3} \log f \\
\geq &~\left(\frac{1}{2} c_b (1-\rho)^{-2/3} + \frac{1}{2}c_b (1-\rho)^{-2/3} - c_d\right) (1-\rho)^{-4/3} \log f \\
\geq &~\left(\frac{1}{2} c_b (1-\rho)^{-4/3}\right) (1-\rho)^{-4/3} \log f \\
 = &~\frac{c_b}{2} (1-\rho)^{-2} \log f = \frac{1}{2} b,
\end{align*}
where the last inequality follows from the fact that $(1-\rho)^{-2/3} \geq 2c_d/c_b$. 
We also have 
\begin{equation}
\frac{1}{2} b = \frac{c_b}{2} (1-\rho)^{-2} \log f \geq \frac{30}{2} \log 3 \geq 1, 
\end{equation}
by the conditions \eqref{eq:const_cond} and \eqref{eq:nrho_cond}.
\end{proof}
\begin{lemma}\label{lem:ell}
Recall the description of the scheduling policy in the lower envelope phase, 
the subintervals $I_0, \cdots, I_{\ell}$ defined in Eq. \eqref{eq:sint} and $\ell$ 
defined by Ineq. \eqref{eq:ell_cond}. 
Then, $\ell$ is well-defined, in that there exists a unique $\ell$ that satisfies \eqref{eq:ell_cond}. 
Furthermore, 
\begin{equation}\label{eq:ell_bound}
\ell \leq \frac{\sqrt{c_d}}{38} \cdot \frac{1}{(1-\rho)^{2/3}},
\end{equation}
and for $u \in \{0, 1, 2, \cdots, \ell-1\}$, 
\begin{equation}
I_u \geq \frac{1}{2}d.
\end{equation}
\end{lemma}
\begin{proof}
Let $\bar{\ell} = \sqrt{c_d} (1-\rho)^{-2/3}/38$, 
which, similar to $b$, $d$ and $s$, we treat as an integer. 
Let us just note that $\bar{\ell}$ is well-defined, in that 
\[
\bar{\ell} = \sqrt{c_d} (1-\rho)^{-2/3}/38 \geq \frac{38}{\sqrt{c_d}} \cdot \frac{\sqrt{c_d}}{38} = 1,
\]
using the condition $(1-\rho)^{-2/3} \geq 38/\sqrt{c_d}$ from \eqref{eq:nrho_cond}. 

For $u \in \left\{0, 1, \cdots, \bar{\ell}\right\}$, let 
\begin{equation}\label{eq:sint1}
I'_u = d - 19 u \sqrt{d \log f}.
\end{equation}
We will show that 
\begin{itemize}
\item[(a)] $I'_u \geq d/2$ for all $u \in \left\{0, 1, \cdots, \bar{\ell}\right\}$; and 
\item[(b)] $I'_0 + I'_1 + \cdots + I'_{\bar{\ell}} \geq b$; 
\end{itemize}
To establish part (a), note that in view of the definition of $c_d$ given in \eqref{eq:policy_d},
\begin{equation}\label{eq:b_ell}
\bar{\ell} = \frac{d}{38\sqrt{d\log f}} = \frac{1}{2} \cdot \frac{d}{19\sqrt{d\log f}}.
\end{equation}
Thus, for $u \in \left\{0, 1, \cdots, \bar{\ell}\right\}$, 
\begin{align*}
I'_u \geq I'_{\bar{\ell}} & = d - 19\bar{\ell} \sqrt{d\log f} \\
& = d - \frac{1}{2} \cdot \frac{d}{19\sqrt{d\log f}} \cdot 19 \sqrt{d\log f} = \frac{1}{2}d.
\end{align*}
Next, consider part (b). We have 
\begin{align*}
\sum_{u=0}^{\bar{\ell}} I'_u \geq &~ \frac{1}{2}\bar{\ell}d 
= \frac{\sqrt{c_d}}{38} (1-\rho)^{-2/3} \cdot \frac{c_d}{2} (1-\rho)^{-4/3} \log f \\
= &~ \frac{c_d^{3/2}}{76} (1-\rho)^{-2} \log f \\
\geq &~c_b (1-\rho)^{-2} \log f = b,
\end{align*}
where the first inequality follows from the fact the $I'_u \geq d/2$ for all $u$, 
and the last inequality follows from the condition \eqref{eq:const_cond}. 

Having established parts (a) and (b), it is now easy to see that 
$\ell$, as characterized by the condition \eqref{eq:ell_cond}, is well-defined and unique, 
that $\ell \leq \bar{\ell} = \frac{\sqrt{c_d}}{38} (1-\rho)^{-2/3}$, 
and that for $u \in \{0, 1, \cdots, \ell-1\}$, $I_u = I'_u \geq d/2$. 
This concludes the proof of the lemma.
\end{proof}
\begin{lemma}\label{lem:bcp}
The length $b-s$ of the backlog clearing phase satisfies 
\begin{equation}\label{eq:bcp}
b-s = c_r (1-\rho)^{-1} \log f,
\end{equation}
where $c_r = c_b - \sqrt{c_s c_b} \geq 1$. In particular, $b-s\geq 1$. 
\end{lemma}
Lemma \ref{lem:bcp} was established in a similar way to Lemma 5.1 in \cite{STZ16}, whose proof is included here for completeness.
\begin{proof}
In view of the definition of $b$ in \eqref{eq:policy_pm}, we have that
\begin{align*}
b-s = &~b-\rho b - \sqrt{c_s b \log f} \\
= &~c_b (1-\rho)^{-1} \log f - \sqrt{c_s c_b (1-\rho)^{-2} \log^2 f} \\
= &~(c_b - \sqrt{c_s c_b}) (1-\rho)^{-1} \log f = c_r (1-\rho)^{-1} \log f. 
\end{align*}
The fact that $c_r \geq 1$ follows from the condition $c_b - \sqrt{c_s c_b} \geq 1$ in \eqref{eq:const_cond}.
\end{proof}
\begin{lemma}\label{lem:ncp}
The length $d+s-b$ of the normal clearing phase satisfies 
\begin{equation}\label{eq:ncp}
d+s-b \geq c_o (1-\rho)^{-4/3} \log f,
\end{equation}
where $c_o = c_d - c_r \geq 1$. In particular, $d+s-b \geq 1$.
\end{lemma}
\begin{proof} We have
\begin{align*}
d+s-b = &~d - (b-s) \\
= &~c_d (1-\rho)^{-4/3} \log f - c_r (1-\rho)^{-1} \log f \\
\geq &~ (c_d - c_r) (1-\rho)^{-4/3} \log f \\
= &~ c_o (1-\rho)^{-4/3} \log f.
\end{align*}
By Lemma \ref{lem:bcp} and condition $c_d \geq c_b \geq c_s \geq 30$ in view of \eqref{eq:const_cond}, we have
\[
c_o = c_d - c_r = c_d - c_b+ \sqrt{c_s c_b} \geq \sqrt{c_s c_b} \geq c_s \geq 30.
\]
\end{proof}
\section{Policy Analysis}\label{sec:analysis}
In this section, we provide a thorough analysis of the performance of our policy, and prove 
the main result of the paper, Theorem \ref{theorem:main}. Our goal is to show that any at time, 
the expected total queue size is upper bounded by $O(nd)$, which immediately leads to Theorem \ref{theorem:main}. 
To do so, we consider what happens during the $k$th arrival period and the $k$th service period, 
and prove the following: 
\begin{itemize}
\item[(i)] During the first $d$ slots of the arrival period, the expected total number of arrivals is 
$O(nd)$.
\item[(ii)] With high probability, during the lower envelope phase, the policy uses full matchings for schedules, 
never idles, and does not waste any offered services (Lemma \ref{lem:no_waste}). 
As a result, the total queue size does not increase in expectation, during this phase.
\item[(iii)] With high probability, all $k$th-batch arrivals can be cleared by the end of 
the normal clearing phase, and the number of new backlogged packets $U_k = 0$ (Lemma \ref{lem:new_backlog}). 
Then, by applying Kingman's bound (Theorem \ref{theorem:discrete_kingman}) to the number of backlogged packets, $B_k$, 
we have $\EE[B_k]\leq 1$ for all $k$ (Lemma \ref{lem:backlog_bound}). 
\end{itemize}
By properly accounting for different components that contribute to the total queue size at any given time, 
and using the facts (i) - (iii), we can establish the upper bound $O(nd)$ on the expected total queue size 
(see Section \ref{ssec:analysis}). 

\subsection{No waste or idling during the lower envelope phase}\label{ssec:no_waste}
Let $k \in \ZZ_+$, and recall from Section \ref{sec:policy} the policy description in the $k$th lower envelope phase. 
Note that in each time slot of this phase, the policy either uses a full matching or idles. 
The aim of the policy in this phase is to serve as many packets as possible. 
Since the phase consists of $b-d$ time slots, and at most $n$ packets can be served in each slot, 
at most $n(b-d)$ packets can be served in this phase. 
There are two ways in which the policy may fail to serve $n(b-d)$ packets in this phase: 
\begin{itemize}
\item[(i)] a service is offered but there is no packet available to serve, in which case 
the corresponding offered service is {\em wasted}; or
\item[(ii)] the policy {\em idles} for at least one slot.
\end{itemize}
Let $\cW^k$ be the event that during the $k$th lower envelope phase, 
the policy fails to serve $n(b-d)$ packets. The following lemma states that 
with high probability, $\cW^k$ does not happen. 
\begin{lemma}\label{lem:no_waste}
We have
\begin{equation}
\PP(\cW^k) \leq \frac{1}{2f^{13}}, \qquad \text{for all } k.
\end{equation}
\end{lemma}
\begin{proof}
Let $k \in \ZZ_+$. By the discussion preceding Lemma \ref{lem:no_waste}, 
we see that $\cW^k$ is also precisely the event that there is either waste 
or idling in the $k$th lower envelope phase. 
Recall that the phase is divided into consecutive subintervals $I_1, I_2, \cdots, I_{\ell}$. 
For $u \in [\ell]$, let $\cW^k_u$ denote the event that there is either waste or idling 
in subinterval $I_u$. 
Then, it is easy to see that 
\begin{equation}
\cW^k = \cW^k_1 \cup \cW^k_2 \cup \cdots \cup \cW^k_{\ell}.
\end{equation}
We now claim that for sufficiently large $n$, 
\begin{equation}\label{eq:no_waste1}
\PP\left(\cW^k_u\right) \leq \frac{1}{2f^{15}}, \quad \text{for all } u \in [\ell], 
\end{equation}
from which it follows immediately that, by union bound, 
\begin{equation}
\PP\left(\cW^k\right) \leq \sum_{u=1}^{\ell} \PP\left(\cW^k_u\right) 
\leq \frac{\ell}{2f^{15}} \leq \frac{f^2}{2f^{15}} = \frac{1}{2f^{13}}, 
\end{equation}
%\jx{Do you use $
%\ell \leq \frac{\sqrt{c_d}}{38} f_n^{2/3}
%$ given in \eqref{eq:ell_bound} in the last inequality?}
establishing the lemma. Here, the last inequality follows from Ineq. \eqref{eq:ell_bound} in Lemma \ref{lem:ell}
and condition $(1-\rho)^{-2/3} \geq \sqrt{c_d}/38$ in \eqref{eq:nrho_cond}, so that
\[
\ell \leq \frac{\sqrt{c_d}}{38} (1-\rho)^{-2/3} \leq (1-\rho)^{-4/3} \leq f^2.
\]
The rest of the proof is devoted to the proof of the claim.

%Let $n$ be sufficiently large so that $n\geq c_d^{3/2}$, and 
By the conditions in \eqref{eq:nrho_cond}, it is not difficult to see that 
\begin{equation}
\rho d \geq 304 \log f.
\end{equation} 
Let $u \in [\ell]$, and recall from Section \ref{sec:policy} that $\bA^{(u-1)} = \left(A^{u-1}_{ij}\right)_{i,j=1}^n$ 
denotes the matrix of arrivals in subinterval $I_{u-1}$. 
Thus, $A_{ij}^{(u-1)}$ are i.i.d binomial random variables with parameters $I_{u-1}$ and $\rho/n$. 
By Lemma \ref{lem:ell}, $I_{u-1} \geq d/2$, so 
\begin{equation}
\rho I_{u-1} \geq \rho d/2 \geq 152 \log f.
\end{equation}
Thus, by Theorem \ref{theorem:rand_low_env}, with probability $1-1/f^{16}$, 
$\bA^{(u-1)}$ has a $\beta$-lower envelope with 
\begin{equation}
\beta \geq \rho I_{u-1} - \sqrt{304 \rho I_{u-1} \log f} -1.
\end{equation}
Then, we have
\begin{align}
\beta \geq &~\rho I_{u-1} - \sqrt{304 \rho I_{u-1} \log f}-1 \nonumber \\
= &~I_{u-1} - (1-\rho) I_{u-1} - \sqrt{304 \rho I_{u-1} \log f}-1 \nonumber \\
\geq &~I_{u-1} - (1-\rho) d - \sqrt{304 d \log f}-1 \nonumber \\
= &~I_{u-1} - c_d (1-\rho)^{-1/3} \log f - \sqrt{304 c_d (1-\rho)^{-4/3} (\log f)^2}-1 \nonumber \\
= &~I_{u-1} - c_d(1-\rho)^{-1/3} \log f - \sqrt{304 c_d} (1-\rho)^{-2/3} \log f-1 \nonumber \\
\geq &~I_{u-1} - \sqrt{c_d} (1-\rho)^{-2/3} \log f - 18 \sqrt{c_d} (1-\rho)^{-2/3} \log f \nonumber \\
= &~I_{u-1} - 19 \sqrt{c_d} (1-\rho)^{-2/3} \log f = I_{u-1} - 19 \sqrt{d\log f} \nonumber \\
= &~d - 19 u \sqrt{d\log f} = I_u.
\end{align}
Here, the second inequality follows from the fact that $I_{u-1} \leq d$ and $\rho \leq 1$, 
and the third inequality follows from the condition $(1-\rho)^{-2/3} \geq c_d$ of \eqref{eq:nrho_cond}, 
and the condition \eqref{eq:const_cond}. 

So far, we have established that for any $u \in [\ell]$, with probability $1-1/f^{16}$, 
$\bA^{(u-1)}$ has an $I_u$-lower envelope. Under this event, during subinterval $I_u$, 
we can find $I_u$ full matchings to serve packets from $\bA^{(u-1)}$ without any waste; 
this event is exactly the complement of $\cW^k_u$.
Thus, for sufficiently large $n$, for any $u \in [\ell]$, 
\begin{equation}
\PP\left(\left(\cW^k_u\right)^c\right) \geq 1 - \frac{1}{f^{16}} \geq 1 - \frac{1}{2f^{15}}, 
\end{equation}
establishing Ineq. \eqref{eq:no_waste1}, and hence the claim and the lemma.
\end{proof}

\subsection{Backlog analysis}\label{ssec:backlog}
Define $\bA^k(b) = \left(A^k_{ij}(b)\right)_{i,j=1}^n$ the matrix of $k$th-batch arrivals by
\begin{equation}\label{eq:kth_arr}
A^k_{ij}(b) = A_{ij}\left((k+1)b\right) - A_{ij}(kb),
\end{equation}
so that $A^k_{ij}(b)$ is the number of arrivals to queue $(i, j)$ in the $k$th batch. Let
\begin{equation}
R^k_i = \sum_{j'=1}^n A^k_{ij'}(b), \qquad C^k_j = \sum_{i'=1}^n A^k_{i'j}(b)
\end{equation}
be the row and column sums for the $k$th-batch arrivals. Define events
\begin{equation}\label{eq:rc_kth}
\cR^k_i = \{R^k_i > s\}, \qquad \cC^k_j = \{C^k_j > s\}, 
\end{equation}
and 
\begin{equation}
\cE^k = \left(\cR^k_1 \cup \cR^k_2 \cup \cdots \cup \cR^k_n \right) 
\cup \left(\cC^k_1 \cup \cC^k_2 \cup \cdots \cup \cC^k_n \right).
\end{equation}
\begin{lemma}\label{lem:exceed}
We have that 
\begin{equation}
\PP\left(\cE^k\right) \leq \frac{1}{2f^{13}}, \qquad \text{for all } k.
\end{equation}
\end{lemma}
Lemma \ref{lem:exceed} is similar to Lemma 6.3 of \cite{STZ16}, 
whose proof we include for completeness.
\begin{proof} 
Consider the event $\cR^k_1 = \{R^k_1 > s\}$. 
Note that $\EE\left[R^k_1\right] = \rho b$. By Ineq. \eqref{eq:bin_conc1} of Theorem \ref{theorem:bin_tail}, 
we have 
\begin{align*}
\PP\left(R^k_1 > s\right) & = \PP\left(R^k_1 > \rho b + \sqrt{c_s b\log f}\right) \\
& = \PP\left(R^k_1 > \EE\left[R^k_1\right] + \sqrt{c_s b\log f}\right) \\
& \leq \exp\left(-\frac{c_s b \log f}{2(\rho b+ x/3)}\right),
\end{align*}
where $x = \sqrt{c_s b\log f}$. Also note that by Lemma \ref{lem:bcp},
\[
\rho b+ x/3 \leq \rho b+ x  = \rho b +  \sqrt{c_s b\log f} = s \leq b. 
\]
Thus, 
\[
\PP\left(R^k_1 > s\right) \leq \exp\left(-\frac{c_s b \log f}{2b}\right) = \frac{1}{f^{c_s/2}} \leq \frac{1}{4f^{14}},
\]
where in the last inequality, we used the condition $n \geq 4$ of \eqref{eq:nrho_cond} 
and $c_s \geq 30$ of \eqref{eq:const_cond}. 
We have similar inequalities for other events $\cR^k_i$ and $\cC^k_j$, 
and by a simple union bound, we have
\[
\PP\left(\cE^k \right) \leq \frac{2n}{4 f^{14}} \leq \frac{1}{2 f^{13}}.
\]
\end{proof}

\begin{lemma}\label{lem:new_backlog}
Let $k \in \ZZ_+$ be fixed. Then, the following hold.
\begin{itemize}
\item[(a)] $U_k = 0$ on any sample path where neither $\cW^k$ nor $\cE^k$ occurs. 
\item[(b)] $\PP(U_k > 0)\leq 1/f^{13}$.
\item[(c)] On any sample path, we have $U_k \leq n^2 b$.
\end{itemize}
\end{lemma}
Note that Lemma \ref{lem:new_backlog} is similar to Lemma 6.4 of \cite{STZ16}. 
We provide the proof of the lemma here for completeness. 
\begin{proof}
(a) Recall the matrix $\bA^k(b)$ of $k$th-batch arrivals defined in Eq. \eqref{eq:kth_arr}. 
Let $\hat{\bQ}^k = \left(\hat{Q}^k_{ij}\right)_{i,j=1}^n$ be the matrix 
of $k$th-batch arrivals that remain by the end of the lower envelope phase, 
and let 
\begin{equation}\label{eq:rc_kth_rem}
\hat{R}_i = \sum_{j'=1}^n \hat{Q}^k_{ij'}, \qquad \hat{C}_j = \sum_{i'=1}^n \hat{Q}^k_{i'j}
\end{equation}
be the row and column sums of $\hat{\bQ}^k$.

Suppose that $\cW^k$ does not occur. 
Then, during the lower envelope phase, the policy served $n(b-d)$ packets from $\bA^k(b)$, 
and each row/column sum decreases by exactly $b-d$ by the end of the phase. 
More precisely, 
\begin{equation}
\hat{R}_i = R_i - (b-d), \text{ and } \hat{C}_j = C_j - (b-d), \text{ for all } i, j \in [n], 
\end{equation}
where we recall that $R_i$ and $C_j$ are the row and column sums of $\bA^k(b)$ (cf. Eq. \eqref{eq:rc_kth}). 

Suppose that $\cE^k$ does not occur. Then, $R_i \leq s$ for all $i$ and $C_j \leq s$ for all $j$. 
Thus, if neither $\cW^k$ nor $\cE^k$ occurs, then
\begin{equation}
\hat{R}_i \leq s-(b-d) = d+s-b, \text{ and } \hat{C}_j \leq d+s-b, \text{ for all } i, j.
\end{equation}
But $d+s-b$ is the length of the normal clearing phase, so all arrivals 
from the $k$th batch will be cleared by the end of the normal clearing phase, 
using the optimal clearing policy described after Theorem \ref{theorem:up_env}, 
and then $U_k = 0$. %\jx{Do you use the Optimal Clearing Theorem here?}

(b) By part (a), we have
\[
\PP(U_k > 0) \leq \PP\left(\cW^k \cup \cE^k \right) \leq \frac{1}{f^{13}}.
\]

(c) The number of packets that can get backlogged from any batch cannot be more than 
the total number of arrivals from that batch, which is upper bounded by $n^2b$. 
\end{proof}
\begin{lemma}\label{lem:backlog_bound}
We have that $\EE[B_k] \leq 1$ for all $k$.
\end{lemma}
Lemma \ref{lem:backlog_bound} is similar to Lemma 6.5 of \cite{STZ16}, 
and we only provide a sketch of its proof here. 
Recall the dynamics of the backlogged packets in Eq. \eqref{eq:backlog_dyn}, so that 
\[
B_{k+1} = \left(B_k+U_k - (b-s)\right)^+ \leq \left(B_k+U_k - 1\right)^+,
\]
where the inequality follows from Lemma \ref{lem:bcp}. 
Then, we can upper bound each $B_k$ sample-path-wise by $\hat{B}_k$, 
where $\hat{B}_0 = 0$, and
\[
\hat{B}_{k+1} = \left(\hat{B}_k+U_k - 1\right)^+.
\]
Using the notation of Theorem \ref{theorem:discrete_kingman} for $\hat{B}_k$, 
we have $\mu = 1, m_{2y} = 1$. We can also show that
\begin{align}
\lambda = \EE\left[U_k\right] \leq f^{-7}, \text{ and } m_{2x} = \EE\left[U_k^2\right] \leq f^{-1}.
\label{eq:U_mean}
\end{align}
Finally, a direct application of Theorem \ref{theorem:discrete_kingman} gives 
\[
\EE\left[B_k\right] \leq \EE\left[\hat{B}_k\right] \leq \frac{f^{-1}+1}{2(1-f^{-7})} \leq 1.
\]
%for sufficiently large $n$.
%[ADD PROOF SKETCH OR FULL PROOF?]

\subsection{Queue size analysis}\label{ssec:analysis}
In this subsection, we prove Theorem \ref{theorem:main}, the main result of this paper. 
The analysis is similar to the one in \cite{STZ16}. We fix a time slot $\tau$ and 
consider two cases, depending on whether $\tau$ is in a lower envelope phase 
or not.

To facilitate the analysis, we introduce the following notation; 
the same notation was also used in \cite{STZ16}. 
Let $t \in [b+1]$, which will be used to index the $b$ slots of 
the $k$th arrival period and the first slot of the $k$th normal clearing phase. 
For $t \in [b]$, let $A^k_{ij}(t)$ denote the number of arrivals to queue $(i, j)$ 
during the first $t$ slots of the $k$th arrival period, so 
\begin{equation}
A^k_{ij}(t) = A_{ij}(kb+t) - A_{ij}(kb).
\end{equation}
Note that this notation is consistent with that introduced earlier, 
in Eq. \eqref{eq:kth_arr} of Section \ref{ssec:backlog}.
Next, for $t \in [b]$, let $S^k_{ij}(t)$ denote the number of $k$th-batch packets 
that arrive in queue $(i, j)$ and get served during the first $t$ slots of the $k$th batch. 
Finally, for $t \in [b]$ ($t = b+1$, respectively), let $Q^k_{ij}(t)$ be the number of $k$th-batch packets 
in queue $(i, j)$ at the beginning of the $t$-th slot of the $k$th arrival period 
(the first slot of the $k$th normal clearing phase, respectively). Then, we have
\begin{equation}\label{eq:q_batch}
Q^k_{ij}(t+1) = A^k_{ij}(t) - S^k_{ij}(t).
\end{equation}

\paragraph{Queue size during the lower envelope phase.} 
Consider time slot $\tau$ that belongs to the $k$th lower envelope phase, 
so that $kb+d+1 \leq \tau \leq (k+1)b$, 
and consider the total queue size $\sum_{i,j=1}^n Q_{ij}(\tau+1)$. 
This queue size includes two components: backlogged packets 
from previous service periods, and $k$th-batch packets. 

By Lemma \ref{lem:backlog_bound}, the number $B_k$ of backlogged packets 
satisfies 
\begin{equation}\label{eq:backlog_bound}
\EE[B_k]\leq 1.
\end{equation} 
Next, consider the $k$th-batch packets that contribute 
to the queue size $\sum_{i,j=1}^n Q_{ij}(\tau+1)$. 
Let $t = \tau - kb$, and 
recall the definition of $Q^k_{ij}(t+1)$ from Eq. \eqref{eq:q_batch}. 
Then, 
\begin{equation}\label{eq:q_ident}
\sum_{i,j=1}^n Q_{ij}(\tau+1) = B_k + \sum_{i,j=1}^n Q^k_{ij}(t+1).
\end{equation}
Also recall the definitions of $A^k_{ij}(t)$ and $S^k_{ij}(t)$. First, we have
\begin{equation}
\EE\left[\sum_{i,j=1}^n A^k_{ij}(t)\right] = \frac{\rho}{n}\cdot n^2 t = \rho n t.
\end{equation}
Second, recall the definition of event $\cW^k$ in Section \ref{ssec:no_waste}. 
If event $\cW^k$ does not occur, 
\begin{equation}
\sum_{i,j=1}^n S^k_{ij}(t) = n(t-d).
\end{equation}
Thus, by Lemma \ref{lem:no_waste}, 
\begin{align*}
\EE\left[\sum_{i,j=1}^n S^k_{ij}(t)\right] & \geq n(t-d) (1-\PP(\cW^k)) \\
& \geq \left(1-\frac{1}{f^{13}}\right) n(t-d) 
\geq \rho n(t-d).
\end{align*}
Thus, 
\begin{align}
\EE\left[\sum_{i,j=1}^n Q^k_{ij}(t+1) \right] = &~\EE\left[\sum_{i,j=1}^n A^k_{ij}(t)\right] - \EE\left[\sum_{i,j=1}^n S^k_{ij}(t)\right] \nonumber \\
\leq &~ \rho n t - \rho n(t-d) \nonumber \\
= &~ \rho n d \leq nd, \quad t = d+1, \cdots, b. \label{eq:q_batch_bd1}
\end{align}
By Eq. \eqref{eq:q_ident}, Ineq. \eqref{eq:backlog_bound} and \eqref{eq:q_batch_bd1}, 
\begin{equation}
\EE\left[\sum_{i,j=1}^n Q_{ij}(\tau+1)\right] = \EE\left[B_k + \sum_{i,j=1}^n Q^k_{ij}(t+1)\right] 
\leq 1 + nd \leq 2nd.
\end{equation}

\paragraph{Queue size outside the lower envelope phase.} 
This part of the analysis is the same as that in Section 6.4 of \cite{STZ16}, 
and is provided here for completeness. 
Let time slot $\tau$ belong to either the $k$th normal clearing phase, 
or the $k$th backlog clearing phase, so that $(k+1)b+1 \leq \tau \leq (k+1)b+d$, 
and consider the total queue size $\sum_{i,j=1}^n Q_{ij}(\tau+1)$. 
This queue size includes three components: backlogged packets, 
$k$th batch packets, and new arrivals from the $(k+1)$st batch. 

The backlogged packets may be from previous service periods, whose number is given by $B_k$, 
or from the $k$ th batch, whose number is upper bounded by $U_k$. 
Thus, in view of \eqref{eq:U_mean} and \eqref{eq:backlog_bound},
the expected total number of backlogged packets is upper bounded by 
$\EE[B_k + U_k]\leq 2$. 

Next, consider the $k$th batch packets that contribute to $\sum_{i,j=1}^n Q_{ij}(\tau+1)$. 
The number of these packets is largest at the beginning of slot $(k+1)b+1$, 
since there are no $k$th batch arrivals from time slot $(k+1)b+1$ onwards. 
Thus, the expected value of the number of these packets is upper bounded by
\[
\EE\left[\sum_{i,j=1}^n Q^k_{ij}(b+1)\right] \leq nd,
\]
where the inequality follows from Ineq. \eqref{eq:q_batch_bd1} for $t = b$.

Finally, the number of arrivals from the $(k+1)$st batch is largest at time $(k+1)b+d$, 
the expectation of which is upper bounded by
\[
\EE\left[\sum_{i,j=1}^n A^{k+1}_{ij}(d)\right] = \frac{\rho}{n}\cdot n^2 d = \rho n d \leq nd. 
\]
Therefore, the expected total queue size $\EE\left[\sum_{i,j=1}^n Q_{ij}(\tau+1)\right]$ is upper bounded by 
$2+nd + nd \leq 3nd$.

\section{Discussion}\label{sec:discussion}
In this paper, we proposed a new scheduling policy for input-queued switches when all arrival rates are equal.
In the regime $1-\rho = \Theta(n^{-1})$,
our policy gives an upper bound of $O\left(n^{7/3}\log n\right)$ on the expected total queue size, 
improving the state-of-the-art bound $O\left(n^{2.5}\log n\right)$.
Similar to the policy in \cite{STZ16}, ours is of the batching type, 
and starts serving packets much earlier than the time when a whole batch has arrived. 
A central component of our policy is the ability to efficiently serve 
packets without waste, from subintervals that are shorter than a whole batch, 
which makes crucial use of a novel, tight characterization 
of the largest $k$-factor in a random bipartite (multi-)graph, which may be of independent interest. 

In the regime $1-\rho = \Theta(n^{-1})$, 
while our upper bound on the expected total queue size is $O\left(n^{7/3}\log n\right)$, 
the conjectured scaling is $\Theta\left(n^2 \right)$. It remains an important and interesting challenge to see 
whether the gap can be closed. 
Among the class of the batching-type policies considered in our paper and in \cite{STZ16}, 
it appears difficult, if not altogether impossible, to reduce the gap for the following reason. 
On the one hand,  a batching-type policy like ours
needs to wait for enough packets to arrive for a certain level of statistical regularity to emerge. 
On the other hand, the number of arrivals that the policy waits for directly determines 
the queue-size scaling that can be achieved. 
The level of regularity required in our policy seems to be the least stringent, 
given that our characterization of the largest $k$-factor in a random bipartite multigraph is tight. 
Thus, we believe that 
fundamentally new methodological tools 
and design of much more elaborate policies are required 
 to make further progress 
on the queue-size scaling, at least in the regime $1-\rho = \Theta(n^{-1})$.

%\bibliographystyle{abbrv}
%\bibliography{refs}

\end{document}